\documentclass[12pt,reqno]{amsart}
\usepackage{amsthm,amsfonts,amssymb,euscript}

\theoremstyle{plain}
  \newtheorem{theorem}[subsection]{Theorem}
  
  \newtheorem{proposition}[subsection]{Proposition}
  \newtheorem{lemma}[subsection]{Lemma}
  \newtheorem{corollary}[subsection]{Corollary}

\def\bphi{{\boldsymbol{\phi} }}
\def\bpsi{{\boldsymbol{\psi} }}
\def\F{{\mathcal F}}
\def\R{{\mathbb{R}}}
\def\C{{\mathbb{C}}}
\def\I{{\mathcal I}}

\def\A{{\mathcal A}}
\def\B{{\mathcal B}}
\def\P{{\mathcal P}}
\def\Q{{\mathcal Q}}
\def\H{{\mathcal H}}
\def\N{{\mathcal N}}
\def\D{{\mathcal D}}
\def\M{{\mathcal M}}
\def\E{{\mathcal E}}
\def\V{{\mathcal V}}

\def\ch{\mbox{ch} (k)}
\def\chbb{\mbox{chb} (k)}
\def\trigh{\mbox{trigh} (k)}
\def\pb{\overline p}
\def\sb{\overline s}
\def\sab{\overline{s_a}}
\def\seb{\overline{s_e}}
\def\sh{\mbox{sh} (k)}
\def\shh{\mbox{sh}}
\def\chh{\mbox{ch}}
\def\shhb{\overline {\mbox{sh}}}
\def\th{\mbox{th} (k)}
\def\chhb{\overline {\mbox{ch}}}
\def\shbb{\mbox{shb} (k)}
\def\chb{\overline{\mbox{ch}  (k)}}
\def\shb{\overline{\mbox{sh}(k)}}
\def\z{\zeta}
\def\zb{\overline{\zeta}}
\def\FF{{\mathbb{F}}}
\theoremstyle{remark}
  \newtheorem{remark}[subsection]{Remark}
  
\theoremstyle{definition}

\begin{document}

\def\vac{{\big\vert 0\big>}}
\def\fpsi{{\big\vert\psi\big>}}
\def\bphi{{\boldsymbol{\phi} }}
\def\bpsi{{\boldsymbol{\psi} }}
\def\F{{\mathcal F}}
\def\R{{\mathbb{R}}}
\def\C{{\mathbb{C}}}
\def\I{{\mathcal I}}
\def\Q{{\mathcal Q}}
\def\ch{\mbox{\rm ch} (k)}
\def\cht{\mbox{\rm ch} (2k)}
\def\chbb{\mbox{\rm chb} (k)}
\def\trigh{\mbox{\rm trigh} (k)}
\def\p{\mbox{p} (k)}
\def\sh{\mbox{\rm sh} (k)}
\def\sht{\mbox{\rm sh} (2k)}
\def\shh{\mbox{\rm sh}}
\def\chh{\mbox{ \rm ch}}
\def\th{\mbox{\rm th} (k)}
\def\thb{\overline{\mbox{\rm th} (k)}}
\def\shhb{\overline {\mbox{\rm sh}}}
\def\chhb{\overline {\mbox{\rm ch}}}
\def\shbb{\mbox{\rm shb} (k)}
\def\chbt{\overline{\mbox{\rm ch}  (2k)}}
\def\shb{\overline{\mbox{\rm sh}(k)}}
\def\shbt{\overline{\mbox{\rm sh}(2k)}}
\def\z{\zeta}
\def\zb{\overline{\zeta}}
\def\FF{{\mathbb{F}}}
\def\S{{\bf S}}
\def\W{{\bf W}}

 \title[Pair excitations,I]%
{ Pair excitations and the mean field approximation of interacting Bosons, I}

\author{M. Grillakis}
\address{University of Maryland, College Park}
\email{mng@math.umd.edu}

\author{M. Machedon}
\address{University of Maryland, College Park}
\email{mxm@math.umd.edu}

\date{}
\dedicatory{}
\commby{}
\maketitle

\begin{abstract}
In our previous work \cite{GMM1},\cite{GMM2} we introduced a correction to the mean field approximation of interacting Bosons.
This correction describes the evolution of  pairs of particles that leave the condensate and subsequently evolve on
a background formed by the condensate.  In \cite{GMM2} we carried out the analysis assuming that the interactions are
 independent of the number of particles $N$.
Here we consider the
case of stronger interactions. We offer a new transparent derivation for the evolution of pair excitations.
Indeed, we obtain a pair of linear equations describing their evolution.
Furthermore, we obtain apriory estimates independent of the number of particles and use these to
compare the exact with the approximate dynamics.
\end{abstract}

\section{Introduction}
\label{sec:intro}

The purpose of our present work is to investigate certain aspects of the evolution of a large number of indistinguishable
Quantum particles (Bosons) under binary interactions.
If we call $\psi(t,x_{1},x_{2}\ldots x_{N})$ the wavefunction describing the $N$ particles with $x_{j}\in \mathbb{R}^{3}$ the
coordinates for $j=1,2\ldots N$, then $\psi$ satisfies an evolution equation of the form
\begin{equation}\label{classevol}
(1/i)\partial_{t}\psi =\left(H_{1} -N^{-1}V\right)\psi
\end{equation}
where $H_{1}$ is a sum of the from
$
\sum_{j=1}^{N}\Delta_{x_{j}}
$.
The term $V$ models two body interactions of the following general type
\begin{equation}\label{binary}
V:=(1/2)\sum_{x_{j}\not= x_{k}}N^{3\beta}v\big(N^{\beta}(x_{j}-x_{k})\big)
\quad ;\quad 0\leq \beta\leq 1
\end{equation}
where $v \in C^1_0$  is non-negative, spherically symmetric, and decreasing.

In the equation \eqref{classevol} above we consider non-relativistic particles and set $h=2m=1$ for simplicity.
Here and for the rest of this paper we denote $v_{N}:=N^{3\beta}v\big(N^{\beta}\cdot\big)$.
The fact that we consider Bosons means that the wavefunction is invariant under all permutations of the indices $j=1,2\ldots N$ and one would like to
solve the evolution equation under some initial condition at, say $t=0$, $\psi(0,x_{1},x_{2}\ldots x_{n}):=\psi_{0}(x_{1},x_{2}\ldots x_{N})$.
Presently we are interested in the evolution of factorized (or approximately factorized) initial data i.e. we would like to consider
special initial data of the form
\begin{equation}\label{tensorproduct}
\psi_{0}:=\prod_{j=1}^{N}\phi_{0}(x_{j})\ .
\end{equation}
The evolution of \eqref{classevol} with initial data \eqref{tensorproduct} is quite complicated for $N$ large
and one would like to have an effective approximate description of the evolution.
The motivation for this type of problem comes from Bose-Einstein condensation where one considers a large number of identical
(indistinguishable)
particles in a trap. Einstein following ideas of Bose, observed that nonineracting particles in a box undergo a phase transition
at a critical  temperature proportional to ${\rm density}^{2/3}$, so that below this temperature a macroscopic number of particles occupy
the ground state, furthermore at zero temperature all particles condense to the ground state.
It is natural and more realistic to consider interacting particles. Following ideas of Landau a heuristic theory based on the idea
of the mean-field approximation was developed by Gross and Pitaevski \cite{pitaevskii61}, \cite{gross61}.
On a more fundamental level, the problem of a weakly interacting Quantum gas
was taken up in the pioneering work of Lee, Huang and Yang as well as Dyson) \cite{Lee-H-Y}, \cite{D}.

More recent theoretical developments are due to Lieb, Solovej, Yngvanson, Seiringer et. al. see \cite{lieb05} and references therein. In particular, Theorems 6.1,  7.1 in \cite{lieb05}, as well as Appendix C in \cite{E-S-Y1}  strongly suggest that
 that the ground state is well approximated by a tensor product  \eqref{tensorproduct}
where $\phi_{0}$ describes a mean field approximation.

 Let us point out that we can in fact treat more general initial conditions corresponding to the $N$th component of
$e^{\sqrt N A(\phi)}e^{B(k)} \vac$, see section \eqref{sec:Fock-deriv}.

Experimental confirmation of Bose-Einstein condensates was finally achieved \cite{andersonetal95},
using alkali atoms. The reason for the use of alkali atoms is the fact that they contain a single valence electron in the outermost $s$-orbital
(for example 5-$s$ for Rubidium). The other contributions to the total spin comes from the nuclear spin.
If the nuclear isotope is one with odd number of protons and neutrons it will have a net half integer spin. For example
Rubidium 87 has $S=3/2$ from the nucleus. The total spin takes the values $S=1$ or $S=2$. If we prepare the sample so that only
one of these states is present then this will be a gas of identical Bosons. If two different states are present then we should consider it
as a mixture of two different gasses.
Since alkali atoms are complicated composite particles their interactions are not known explicitly, which means that the potential $v$
in
\eqref{binary} is not explicitly known,
moreover one can treat the atoms as Bosons only for a sufficiently dilute gas. At shorter distances the internal structure of the
atoms should be taken into account.   Here
we consider a sufficiently dilute Boson gas and we make the reasonable assumption that the interactions are repulsive i.e.
$v\geq 0$ and that they are short range in the sense that
$\int v(x)dx <\infty$. It is clear from the above comments that one should consider particles with spin. The present framework can be
generalized in this case in a straightforward manner, however in the name of simplicity we forgo this generalization.

Let us comment on the scaling present in the form of binary interactions. The parameter $\beta$ describes the strength of particle interactions.
It is a reasonable (but not obvious) idea to assume that the evolution of \eqref{classevol} is approximated by a tensor product i.e.
\begin{equation}\label{meanfield1}
\psi(t)\approx\prod_{j=1}^{N}\phi(t,x_{j})
\end{equation}
and the issues are, first to explain the nature of the approximation described in \eqref{meanfield1} and second to derive an evolution equation for $\phi(t,x)$ consistent
with the dynamics of \eqref{classevol}.
On the second question,
the general idea is that the evolution of the mean field $\phi(t,x)$  satisfies an equation of the form
\begin{equation}\label{hartreeGP}
i\partial_{t}\phi =\Delta\phi -g_{\beta}\big(\vert\phi\vert^{2}\big)\phi\ ,
\end{equation}
and the nonlinear term $g_{\beta}(\vert\phi\vert^{2})$ depends on $\beta$ in the following manner,
\begin{align*}
&g_{0}(\vert\phi\vert^{2})=\int dy\left\{v(x-y)\vert\phi(t,y)\vert^{2}\right\}
\\
&g_{\beta}(\vert\phi\vert^{2})=\left(\int v(y)dy\right)\vert\phi(t,x)\vert^{2}\quad ;\quad 0<\beta <1
\\
&g_{1}(\vert\phi\vert^{2})=8\pi a\vert\phi(t,x)\vert^{2}\ ,
\end{align*}
where $a$ appearing in $g_{1}$ is the scattering length corresponding to the potential $v$.
In the case $\beta =0$ one obtains a Hartree type evolution for the mean field and considerations similar to our present work where
taken up in \cite{GMM1},\cite{GMM2}.
The case $\beta =1$ is probably the most interesting. In this case the scaling is critical in the sense
that particles develop short scale correlations which in the limit $N\to\infty$ lead to the
appearance of the scattering length in the equation. A heuristic argument for this is well known in the Physics community,
however the explanation on how the scattering length emerges from the $N$ body dynamics was recently given in the work of Erd\"os, Schlein and Yau \cite{E-S-Y3},\cite{E-S-Y4}.

Our aim is to introduce pair excitations as a correction to the mean field approximation. This goal is achieved by introducing a kernel
$k(t,x,y)$ which describes pair excitations and one would like to derive an evolution equation for $k$ consistent with the
$N$ body dynamics,
which means that we should be able to obtain
estimates comparing the exact with the approximate dynamics.
The general idea of the approximation can be described in the following manner. Two particles leave the condensate and
form a pair $v_{N}(x_{1}-x_{2})\phi(x_{1})\phi(x_{2})$ which in turn drives the evolution of pair interactions.
It turns out that a natural way to introduce pair excitations as a correction to the mean field is via a Fock space formalism which we will outline in the next section.
Let us comment here on the nature of our approximation. The mean field approximation \eqref{meanfield1} is a simple description of the $N$-body wavefunction,
however the nature of the approximation is quite involved and uses the BBGKY hierarchy and its limit as $N\to\infty$
as shown by Elgart, Erd\"os, Schlein and Yau  \cite{E-E-S-Y1, E-Y1, E-S-Y1,  E-S-Y2, E-S-Y4, E-S-Y3}.
See  the approach of \cite{K-MMM}, \cite{KSS}, \cite{CP} based on space-time estimates. We also mention
the related case of 3 body interactions \cite{XC1}, and switchable quadratic traps \cite{XC2, XC3}.

Moreover the approximation does not track the
exact dynamics, rather its true usefullness lies in the fact that it can (approximately) track observables. In contrast our approximation is more complicated but it tracks
the exact dynamics in Fock space norm. As a matter of fact a heuristic explanation of our approximation runs as follows:
The $N$-body wave function consists of three parts,  particles that live in the condensate, bound pairs and particles that decayed after forming pairs.
Controlling the number of particles that formed pairs leads to another justification of the mean field approximation.
We will not pursue this line of inquiry here, however the approximation can be readily used to estimate observables.
There are two main points in our present work. First we have a new transparent derivation of the evolution equation of
pair excitations, indeed we derive a new system of linear equations. Second we obtain apriory estimates for the pair excitations kernel which
are independent of $N$ and this, in turn, allows us to estimate the difference between the exact and approximate solutions provided that
$\beta$ is sufficiently small ($\beta < \frac{1}{6}$).

Our work was inspired by \cite{Rod-S} as well as \cite{wuI}. Previous works directly related to the present are \cite{G-V} and \cite{hepp}. See also \cite{Bo} and
\cite{margetisII}.

See Theorem \eqref{mainthm} below for the precise statement of our main result.

The paper is organized as follows. In section 2 we develop the Fock space formalism which is necessary for our computations and derive
the evolution equations for the pair excitation kernel.
In section 2 and 3 we we derive the apriory estimates for the mean-field and for the pair excitation kernel.
In section 4 we show how this information can be implemented in order to compare the exact solution to our approximation.

\section{Fock space formalism and the new derivation}
\label{sec:Fock-deriv}

In this section we introduce the Fock space formalism and the Hamiltonian evolution in symmetric Fock space.
$\FF$ is a Hilbert space consisting of vectors of the form
\begin{equation*}
\big\vert\psi\big>=\big(\psi_{0}\ ,\ \psi_{1}(x_{1})\ ,\ \psi_{2}(x_{1},x_{2})\ ,\ \ldots\ \big)
\end{equation*}
where $\psi_{0}\in \C$ and $\psi_k$ are symmetric $L^2$ functions. The norm of such a vector is,
\begin{equation*}
\big\Vert\ \big\vert\psi\big>\big\Vert^{2}=\big<\psi\big\vert\psi\big> =\vert\psi_{0}\vert^{2}+
\sum_{n=1}^{\infty}\big\Vert\psi_{n}\big\Vert^{2}_{L^{2}}\ .
\end{equation*}
Thus $\FF$ is a direct sum of sectors $\FF_{n}$ of the form,
\begin{equation*}
\FF =\sum_{n=0}^{\infty}\FF_{n}\quad ;\quad
\FF_{n}:=L^{2}_{s}\big(\R^{3n}\big)
\end{equation*}
with $\FF_{0}=\C$ and $L^{2}_{s}(\R^{3})$ denoting the subspace of symmetric functions.
In the Fock space $\FF$ we introduce creation and anihilation distribution valued operators denoted by $a^{\ast}_{x}$ and
$a_{x}$ respectively which act on sectors $\FF_{n-1}$ and $\FF_{n+1}$
in the following manner,
\begin{align*}
&a^{\ast}_{x}(\psi_{n-1}):=\frac{1}{\sqrt{n}}\sum_{j=1}^{n}\delta(x-x_{j})
\psi_{n-1}(x_{1},\ldots ,x_{j-1},x_{j+1},\ldots ,x_{n})
\\
&a_{x}(\psi_{n+1}):=\sqrt{n+1}\psi_{n+1}([x],x_{1},\ldots ,x_{n})
\end{align*}
with $[x]$ indicating that the variable $x$ is frozen. In addition  $a_{x}$ kills $\FF_{0}$ i.e. $a_{x}(\psi_{0})=0$.
The vacuum state will play an important role later and we define it as follows
\begin{equation*}
\vac :=(1,0,0\ldots )
\end{equation*}
so that $a_{x}\vac =0$. One can easily check that $\big[a_{x},a^{\ast}_{y}\big]=\delta(x-y)$ and since
the creation and anihilation operators are distribution valued we can form operators that act on $\FF$ by introducing
a field, say $\phi(x)$, and form
\begin{align*}
a_{\bar{\phi}}:=\int dx\left\{\bar{\phi}(x)a_{x}\right\}
\quad {\rm and}\quad
a^{\ast}_{\phi}:=\int dx\left\{\phi(x)a^{\ast}_{x}\right\}
\end{align*}
where by convention we associate $a$ with $\bar{\phi}$ and  $a^{\ast}$ with $\phi$.
These operators are well defined, unbounded, on $\FF$ provided that $\phi$ is square integrable.
The creation and anihilation operators provide a way to introduce coherent states in $\FF$ in the following manner,
first define
\begin{equation*}
\A(\phi):=\int dx\left\{\bar{\phi}(x)a_{x}-\phi(x)a^{\ast}_{x}\right\}
\end{equation*}
and then introduce $N$-particle coherent states as
\begin{equation}\label{cohstates}
\big\vert\psi(\phi)\big>:=e^{-\sqrt{N}\A(\phi)}\vac \ .
\end{equation}
It is easy to check that
\begin{equation*}
e^{-\sqrt{N}\A(\phi)}\vac =\left(\ldots\ c_{n}\prod_{j=1}^{n}\phi(x_{j})\ \ldots\right)
\quad {\rm with}\quad  c_{n}=\big(e^{-N}N^{n}/n!\big)^{1/2}\ .
\end{equation*}
In particular, by Stirling's formula, the main term that we are interested in has the coefficient
\begin{align}
c_{N}\approx (2\pi N)^{-1/4} \label{stirling}
\end{align}
Thus a coherent state introduces a tensor product in the sector $\FF_{N}$,
hence we can use such states as a mean field approximation to the Hamiltonian evolution in Fock space, see \eqref{meanfield1}.

The Fock Hamiltonian (acting on Fock space vectors) is
\begin{subequations}
\begin{align}
&\H:=\H_{1}-N^{-1}\V\quad\quad\quad   {\rm where,} \label{FockHamilt1-a}\\
&\H_{1}:=\int dxdy\left\{
\Delta_{x}\delta(x-y)a^{\ast}_{x}a_{y}\right\}\quad {\rm and} \label{FockHamilt1-b}\\
&\V:=\frac{1}{2}\int dxdy\left\{v_{N}(x-y)a^{\ast}_{x}a^{\ast}_{y}a_{x}a_{x}\right\}\ ,
\label{FockHamilt1-c}
\end{align}
\end{subequations}
where we set
\begin{equation}\label{scaledpotential}
v_{N}(x-y):=N^{3\beta}v\big(N^{\beta}\vert x-y\vert\big)\ ,
\end{equation}
and the evolution in Fock space is described by the equation,
\begin{equation}\label{fockeq}
\frac{1}{i}\partial_{t}\big\vert\psi\big>
=\H\big\vert\psi\big>
\end{equation}
which has the formal solution
\begin{equation}\label{exactfockevol}
\big\vert\psi(t)\big>=e^{it\H}\big\vert\psi_{0}\big>\ .
\end{equation}
Notice that $\H$ preserves the sectors $\FF_{n}$ and that $\H$ agrees with the classical Hamiltonian \eqref{classevol} on $\FF_{N}$. However in this
framework we allow any number of particles to evolve and one is interested, in particular, in the evolution on the sector $\FF_{N}$.

Our goal is to approximate $\big\vert\psi(t)\big>$ in \eqref{exactfockevol} and for this purpose we introduce two fields $\phi(t,x)$ and
$k(t,x,y)$ and the associated operators,
\begin{subequations}
\begin{align}
&\A(\phi):=\int dx\left\{\bar{\phi}(t,x)a_{x}-\phi(t,x)a^{\ast}_{x}\right\}\label{meanfield-2}
\\
&\B(k):=\int dxdy\left\{\bar{k}(t,x,y)a_{x}a_{y}
-k(t,x,y)a^{\ast}_{x}a^{\ast}_{y}\right\}\ . \label{pairexcit-2}
\end{align}
\end{subequations}

The coherent initial data are introduced via $\big\vert\psi_{0}\big>=e^{-\sqrt{N}\A(\phi_{0})}\big\vert 0\big>$ which means that the initial data
are a tensor product  on $\FF_{N}$ as desired, see \eqref{tensorproduct}.
Our approximation scheme is
\begin{equation}\label{approx}
\big\vert\psi_{appr}\big>:=e^{-\sqrt{N}\A(t)}e^{-\B(t)}e^{iN\chi(t)}\big\vert 0\big>
\end{equation}
with $\chi(t)$ a phase factor,
and we plan to show that $\big\vert\psi(t)\big>\approx \big\vert\psi_{appr}(t)\big>$.

The issue for us is to determine the dynamics of the fields $\phi$ and $k$. It turns out the the evolution of $k$ is described via a set of new fields,
\begin{subequations}
\begin{align}
\sh &:=k + \frac{1}{3!} k \circ \overline k \circ k + \ldots~, \label{hypebolicsine}\\
\ch &:=\delta(x-y) + \frac{1}{2!}\overline k  \circ k + \ldots~,\label{hypebolicosine}
\end{align}
\end{subequations}
where $\circ$ indicates composition, namely $k\circ l$ stands for the product,
\begin{align*}
(k\circ l)(x_{1}x_{2}):=\int dy\left\{k(x_{1},y)l(y,x_{2})\right\}\ .
\end{align*}
A crucial property of the above multiplication is that it is not commutative i.e. $k\circ l\not= l\circ k$.
In order to describe the evolution we need
\begin{subequations}
\begin{align}
&g_{N}(t, x, y):= - \Delta_x \delta (x-y)
 +(v_N * |\phi|^2 )(t, x) \delta(x-y) \nonumber \\
&\qquad\qquad\  +v_N(x-y) \overline\phi (t, x)  \phi(t, y)\label{op-g}
\\
&m_{N}(x, y):=- v_N(x-y) \phi(x) \phi(y) \ .\label{m-term}
\end{align}
\end{subequations}
Using $g_{N}$ we can construct two operators as follows: For a function $s(t,x,y)$ symmetric in $(x,y)$ and a function $p(t,x,y)$
conjugate symmetric in $(x,y)$ i.e. $\bar{p}^{T}=p$, we define
\begin{subequations}
\begin{align}
{\bf S}(s)&:= \frac{1}{i} s_t  + g_{N}^T \circ s + s \circ g_{N} \label{op-S}\\
{\bf W}(p)&:= \frac{1}{i} p_t  +[g_{N}^T, p]\ . \label{op-W}
\end{align}
\end{subequations}
The dynamics of the fields are determined via,
\begin{subequations}
\begin{align}
&\frac{1}{i}\partial_{t}\phi -\Delta\phi +\big(v_{N}\ast\vert\phi\vert^{2}\big)\phi =0 \label{meanfield-3}
\\
&{\bf S}\left(\sht\right)=m_{N} \circ \cht + \chbt \circ m_{N}  \label{pair-S-1} \\
&{\bf W}\left(\chbt\right)= m_{N} \circ \shbt- \sht \circ \overline m_{N}\ . \label{pair-W-1}
\end{align}
\end{subequations}
Recall we assume $v \in C^1_0$  is non-negative, spherically symmetric, and decreasing.
\begin{remark}
It is clear that $\cht$ and $\sht$ are not independent of each other, thus we can ignore the third equation, however in the form stated above the
equations are readily amenable to the derivation of apriori estimates. The equation for $\phi$ is of Hartree type and its formal limit
as $N\to\infty$ is NLS.
\end{remark}

The theorem concerning the evolution of the mean field $\phi$ and the pair excitation kernel $k$ reads as follows.
\begin{theorem} \label{mainestthm} Suppose that $0<\beta <1$ in \eqref{scaledpotential}.
Given initial data $\phi(0,x):=\phi_{0}(x) \in W^{k, 1}$ ($k$ derivatives in $L^1$, with $k$ sufficiently large) and $k(0,x,y):=0$ the system \eqref{meanfield-3}\eqref{pair-S-1}\eqref{pair-W-1}
has global solutions in time which satisfy the apriori estimates,
\begin{subequations}
\begin{align}
&\Vert \phi(t)\Vert_{H^{s}(\R^{3})}\leq C_{s} \label{hartreeest-1}
\\
&\Vert\phi(t)\Vert_{L^{\infty}(\R^{3})}+\Vert \partial_{t}\phi(t)\Vert_{L^{\infty}(\R^{3})}\leq \frac{C}{t^{3/2}} \label{hartreeest-2}
\\
&\Vert \sht(t)\Vert_{L^{2}(\R^{6})}+\Vert\cht(t)-\delta\Vert_{L^{2}(\R^{6})}\leq C\log(1+t)\ . \label{pairestim-1}
\end{align}
\end{subequations}
\end{theorem}


The main difficulty in obtaining the estimates in the theorem above is the fact that $v_{N}$ defined in \eqref{scaledpotential} has a formal limit
$v_{N}(x-y)\to c\delta(x-y)$ which means that $m_{N}$ has a limit which is not square integrable, as a matter of fact it does not belong to
any $L^{p}$ for $p>1$.
In view of the theorem above, we can compare the exact with the approximate evolutions and the result is the following theorem.

\begin{theorem}\label{mainthm}
Suppose that $\big\vert\psi(t)\big>$ is the solution of \eqref{fockeq} with initial data $\big\vert\psi_{0}\big>:=e^{-\sqrt{N}\A(\phi_{0})}\vac$ and
$\big\vert\psi_{appr}(t)\big>$ is the approximation in \eqref{approx} where the evolution of the fields $\phi$ and $k$ is
determined from theorem \eqref{mainestthm}. Under these conditions the following estimate holds,
\begin{equation}\label{mainestim}
\big\Vert\big\vert\psi(t)\big>-\big\vert\psi_{appr}(t)\big>\big\Vert_{\FF}\leq \frac{C (1+t) \log^4(1+t)}{N^{(1-3\beta)/2}}\ .
\end{equation}
provided $0 < \beta < \frac{1}{3}$. This is a meaningful approximation of the $N$th component
of $|\psi \big>$ provided $0 < \beta < \frac{1}{6}$, because of formula \eqref{stirling}. A slightly more precise form of the estimate could be obtained by integrating the right hand side of the inequalities in Proposition \eqref{errorest}.
\end{theorem}

\begin{remark}
The real phase factor $\chi$  is described via $\chi(t):=\int^{t} dt_{1}\big\{\mu_{0}(t_{1})+N^{-1}\mu_{1}(t_{1})\big\}$
where
\begin{align}
\mu_{0}(t)=\frac{1}{2}\int dxdy\left\{
v_{N}(x-y)\vert\phi(t,x)\vert^{2}\vert\phi(t,y)\vert^{2}\right\}\label{mu-1}
\end{align}
and $\mu_{1}$ is a complicated integral given in \eqref{mu-I}
\end{remark}

\begin{proof} Here is an outline of the proof of this theorem.
In order to relate the exact with the approximate solution we introduce the reduced dynamics
\begin{equation}\label{reduceddynamics}
\big\vert\psi_{red}(t)\big> :=e^{\B(t)}e^{\sqrt{N}\A(t)}e^{it\H}e^{-\sqrt{N}\A(0)}\big\vert 0\big>
\end{equation}
i.e. we follow the exact dynamics for time $t$ and then go back following the approximate evolution.
Notice that
$\big\vert\psi_{red}(0)\big>=\big\vert 0\big>$ and if our approximation was following the  exact evolution we would have that
$\big\vert\psi_{red}(t)\big>=\vac$.
Thus our goal is to estimate the deviation of the evolution from the vacuum state. It is straightforward to compute the evolution
of $\big\vert\psi_{red}\big>$ and it is
\begin{equation}
\frac{1}{i}\partial_{t}\big\vert\psi_{red}\big>=
\H_{red}\big\vert\psi_{red}\big>
\end{equation}
where the  (self-adjoint) reduced Hamiltonian is,
\begin{align}\label{reducedHamiltonian}
\H_{red} &:=\frac{1}{i}\big(\partial_{t}e^{\B}\big)e^{-\B}\nonumber
\\
&+e^{\B}\left(\frac{1}{i}\big(\partial_{t}e^{\sqrt{N}\A}\big)e^{-\sqrt{N}\A}
+e^{\sqrt{N}\A}\H e^{-\sqrt{N}\A}\right)e^{-\B}\ .
\end{align}
The main idea is that the evolution of the fields $\phi$ and $k$ is chosen so that the reduced Hamiltonian looks like
\begin{equation*}
\H_{red}=N\mu(t)+\int dxdy \left\{L(t,x,y)a^{\ast}_{x}a_{y}\right\}
-N^{-1/2}\E(t)
\end{equation*}
where $\E(t)$ is an error term containing polynomials in $(a,a^{\ast})$ up to degree four, and $L$ is some self-adjoint expression which is irrelevant for the
rest of the argument.

Next consider
\begin{equation*}
\big\vert\widetilde{\psi}\big>:=e^{-iN\chi(t)}\left(\big\vert\psi_{red}\big> \right)-\vac
\quad {\rm where}\quad \chi(t):=\int^{t}\mu(t_{1})dt_{1}
\end{equation*}
where we called $\mu :=\mu_{0}+N^{-1}\mu_{1}$.
 Thus
 \begin{equation*}
\left(\frac{1}{i}\partial_{t}- \H_{red} + N \mu(t)  \right)e^{-iN\chi(t)} \big\vert\psi_{red}\big> =0.
\end{equation*}
and therefore
\begin{equation*}
\left(\frac{1}{i}\partial_{t}- \H_{red} + N \mu(t)  \right)\big\vert\widetilde{\psi}\big>
= N^{-1/2}\E(t)\vac\
\end{equation*}
The equation above has a forcing term namely $N^{-1/2}\E (t)\vac$ and a standard energy estimate together with the fact that
$e^{\sqrt{N}\A}$ and $e^{\B}$ are unitary,
 gives
\begin{align}\label{energyestim1}
&\big\Vert\big\vert\psi(t)\big>-\big\vert\psi_{appr}(t)\big>\big\Vert_{\FF} \notag\\
&=\big\Vert\big\vert\widetilde{\psi}(t)\big>\big\Vert_{\FF}\leq
N^{-1/2}\int_{0}^{t} dt_{1}\left\{\big\Vert\E(t_{1})\vac\big\Vert_{\FF}\right\}\ .
\end{align}

The proof will be complete if we estimate the right hand side in the above inequality. Notice that $\E\vac$ has entry only in
Fock sectors $\FF_{j}$ for $j=1,2,3,4$ and in order to estimate it we need the lemma below .
\end{proof}
\begin{lemma}
The error term is described as follows,
\begin{equation*}
\E:=e^{\B}\left(\P_{3}+N^{-1/2}\P_{4}\right)e^{-\B}
\end{equation*}
where $\P_{3}$ and $\P_{4}$ are cubic and quartic polynomials in $(a,a^{\ast})$ respectively.
Moreover the following estimate holds if $0 \le \beta < \frac{1}{3}$,
\begin{equation} \label{finalestim1}
\big\Vert\E(t)\vac\big\Vert_{\FF}\leq CN^{3\beta /2} \log^4(1+t)\ .
\end{equation}
A more precise estimate is given in Proposition \eqref{errorest}.
\end{lemma}

\begin{remark}
The polynomials $\P_{3}$ and $\P_{4}$ appearing in the error term are given by the expressions,
\begin{align*}
\P_{3}:=&\int dxdy\left\{v_{N}(x-y)\big(\phi(y)a^{\ast}_{x}a^{\ast}_{y}a_{x}+\bar{\phi}(y)a^{\ast}_{x}a_{x}a_{y}\big) \right\}
\\
\P_{4}:=&(1/2)\int dxdy\left\{v_{N}(x-y)a^{\ast}_{x}a^{\ast}_{y}a_{x}a_{y}\right\}
\end{align*}
as we will see shortly.
\end{remark}

The rest of this section is devoted to the derivation of \eqref{meanfield-3}\eqref{pair-S-1}\eqref{pair-W-1}.
We have to compute $\H_{red}$ above, see \eqref{reducedHamiltonian}, and for this task there are two crucial ingredients. They are based on the  formal identities below
for any two operators, say $\A$ and $\H$,
\begin{subequations}
\begin{align}
&e^{\A}\H e^{-\A}=\sum_{n=0}^{\infty}\frac{1}{n!}\big({\rm ad}_{\A}\big)^{n}\big(\H\big)\label{adjoined-1}
\\
&\big(\partial_{t}e^{\A}\big)e^{-\A}=\sum_{n=1}^{\infty}\frac{1}{n!}\big({\rm ad}_{\A}\big)^{n-1}\big(\A_{t}\big)
\label{adjoined-2}
\end{align}
\end{subequations}
where ${\rm ad}_{\A}\big(\H\big):=\big[\A,\H\big]$. They indicate that we have to compute  repeated commutators of various
operators. The series defining the exponentials in \eqref{adjoined-1}, \eqref{adjoined-2} converge absolutely on the dense subset of vectors with finitely many
nonzero entries provided that $\A=\A(\phi)$ is a polynomial of degree one with $\phi\in L^{2}$ or $\A=\B(k)$ is second order with $\Vert k\Vert_{L^{2}}$
small. If $\B$ is skew-Hermitian, $e^{\B}$ extends as a unitary operator for all $k\in L^{2}$. This construction is closely related to
the Segal-Shale-Weil representation, as explained in \cite{hepp}, \cite{F}, \cite{shale}, and our appendix \eqref{app}.
This calculation was also used in our previous papers \cite{GMM1, GMM2}.

The first observation is the fact that since $\A(\phi)$ is a degree one polynomial, if we denote by $\P_{n}$ a homogeneous polynomial of degree $n$ then
commuting with $\A$ produces $\big[\A,\P_{n}\big]=\P_{n-1}$ i.e. a homogeneous polynomial of degree $n-1$. This in turn implies that repeated
commutators produce a finite series in \eqref{adjoined-1}, \eqref{adjoined-2} which can be computed explicitly after some tedious but straightforward calculations.

The second observation is that in \eqref{adjoined-1}, \eqref{adjoined-2} when we replace $\A$ with $\B$ we obtain infinite series with a certain periodicity which allows for
explicit summation. This can be expressed via a Lie algebra isomorphism.
For symplectic matrices of the blocked form
\begin{equation*}
L:=\left(\begin{matrix}
d(x,y)&l(x,y)\\
k(x,y)&-d(y,x)
\end{matrix}\right)
\end{equation*}
where $d$, $k$ and $l$ are kernels in $L^{2}$, and $k$ and $l$ are symmetric in $(x,y)$, we define the map from $L$ to
quadratic polynomials is $(a,a^{\ast})$ in the following manner,
\begin{equation}\label{liemap}
\I\big(L\big) =
\frac{1}{2}\int dxdy\left\{
(a_{x}\ ,\ a^{\ast}_{x})\left(\begin{matrix}d(x,y)&l(x,y)\\
k(x,y)&-d(y,x)\end{matrix}\right)\left(\begin{matrix}-a^{\ast}_{y}\\ a_{y}\end{matrix}\right)\right\}\ .
\end{equation}
The crucial property of this map is the Lie algebra isomorphism
\begin{equation}\label{Lieisomorph}
\big[\I(L_{1}),\I(L_{2})\big]=\I\big([L_{1},L_{2}]\big)
\end{equation}
thus any computation that involves commutations can be performed in the realm of symplectic matrices and then transfered
to polynomials in $(a,a^{\ast})$. In particular if we call
$\I(H)=\H$ for a quadratic Hamiltonian and $\I(K)=\B$ then we have the two formulas below,
\begin{subequations}
\begin{align}
e^{\B}\H e^{-\B}&=\I\big(e^{K}He^{-K}\big)\label{formula-1}
\\
\big(\partial_{t}e^{\B}\big)e^{-\B}&=\I\big((\partial_{t}e^{K})e^{-K}\big)\ .
\label{formula-2}
\end{align}
\end{subequations}
Actually, to avoid the infinite trace in \eqref{formula-1}, we write
\begin{align*}
e^{\B}\H e^{-\B}&=
\H + [e^{\B}, \H]e^{-\B}\\
&=\H + \I\big([e^{K}, H]e^{-K}\big)
\end{align*}
As a matter of fact if we define the following quardatic expressions,
\begin{align*}
&\D^{\dag}_{xy}:=a_{x}a^{\ast}_{y}\quad ;\quad \D_{xy}:=a^{\ast}_{x}a_{y}
\\
&\Q^{\ast}_{xy}:=a^{\ast}_{x}a^{\ast}_{y}\quad ;\quad Q_{xy}:=a_{x}a_{y}
\end{align*}
then we can write,
\begin{align*}
\I(L)=-\frac{1}{2}\int dxdy\left\{
d(x,y)\D^{\dag}_{xy}+d(y,x)\D_{xy}+k(x,y)\Q^{\ast}_{xy}-l(x,y)\Q_{xy}\right\}\ .
\end{align*}

\begin{remark}
Notice that $D^{\dag}_{xy}=D_{yx}+\delta(x-y)$ thus we can write
\begin{align*}
\I(L)=&-\frac{1}{2}\int dxdy\left\{
d(x,y)\D_{yx}+d(y,x)\D_{xy}+k(x,y)\Q^{\ast}_{xy}-l(x,y)\Q_{xy}\right\}\\
&-\frac{1}{2}\int dx\{d(x,x)\}
\ .
\end{align*}
\end{remark}

In our present formalism if we define the matrix
\begin{align*}
&K=\left(
\begin{matrix}
0 &\overline{k}\\
k & 0
\end{matrix}\right)\ ,
\end{align*}
then we have that $\I(K)=\B$, see the expression in \eqref{pairexcit-2}. The exponential of $K$ can be computed,
\begin{align*}
e^K &=\left(
\begin{matrix}
\ch &\shb\\
\sh & \chb
\end{matrix}
\right)
\qquad
{\rm where,}
\\
\sh &:=k + \frac{1}{3!} k \circ \overline k \circ k + \ldots~, \\
\ch &:=\delta(x-y) + \frac{1}{2!}\overline k  \circ k + \ldots~,
\end{align*}
and $\circ$ indicates composition. For completeness and for the convenience of the reader we include in the appendix the derivation of
\eqref{Lieisomorph}, see also \cite{GMM1}, \cite{F}.


Let us now proceed with the calculations. First look at the expression inside the parentheses in the reduced Hamiltonian \eqref{reducedHamiltonian}.
It is straightforward after repeated (but finite) commutations with $\A$ to come up with the expression below (see section 3 of \cite{GMM1}),
\begin{align}\label{firstreduction}
&\frac{1}{i}\big(\partial_{t}e^{\sqrt{N}\A}\big)e^{-\sqrt{N}\A}
+e^{\sqrt{N}\A}\H e^{-\sqrt{N}\A} \nonumber
\\
&=N\mu_{0} +N^{1/2}\P_{1}+\P_{2}-N^{-1/2}\P_{3}-N^{-1}\P_{4}
\end{align}
where $\P_{n}$ indicate polynomials of degree $n$ to be given explicitly below. The first term $\mu_{0}$ is a scalar which can be absorbed in
the evolution as an extra phase factor. It is given by the commutators,
\begin{equation*}
\frac{1}{2i}\big[\A,\partial_{t}\A\big]+\frac{1}{2}\big[\A,[\A,\H_1]\big]
-\frac{1}{4!}\Big[\A,\big[\A,[\A,[\A,\V]]\big]\Big]
\end{equation*}
which reduce to the expression below,
\begin{align}
\mu_{0}&:=\int dx\left\{\frac{1}{2i}\big(\phi\bar{\phi}_{t}-\bar{\phi}\phi_{t}\big)-\big\vert\nabla\phi\big\vert^{2}
\right\}\nonumber
\\
&-\frac{1}{2}\int dxdy\left\{v_{N}(x-y)\vert\phi(x)\vert^{2}\vert\phi(y)\vert^{2}\right\}\ . \label{phase1-I}
\end{align}
The first degree polynomial $\P_{1}$ arise from the commutators,
\begin{equation*}
\frac{1}{i}\partial_{t}\A +\big[\A,\H_1\big]-\frac{1}{3!}\big[\A,[\A,[\A,\V]]\big]
\end{equation*}
and it can be expressed as follows,
\begin{equation}\label{meanfield2}
\P_{1}=\int dx\left\{ h(t,x)a^{\ast}_{x}+\bar{h}(t,x)a_{x}\right\}
\end{equation}
where $h:=-(1/i)\partial_{t}\phi +\Delta\phi -\big(v_{N}\ast\vert\phi\vert^{2}\big)\phi$.
The second degree polynomial consists of the terms
\begin{equation*}
\H_{1}-\frac{1}{2}\big[\A,[\A,\V]\big]
\end{equation*}
and
can expressed
\begin{align}
\P_{2}&=\frac{1}{2}\int dxdy\left\{-g_{N}(t,x,y)\D_{yx}-g_{N}(t,y,x)\D_{x,y}\right\} \nonumber
\\
&+\frac{1}{2}\int dxdy\left\{\overline{m}_{N}(t,x,y)\Q_{xy} +
m_{N}(t,x,y)\Q^{\ast}_{x,y}\right\}
\label{2nd-order-I}
\end{align}
where $g_{N}$ and $m_{N}$ are given by, see \eqref{op-g}, \eqref{m-term}
\begin{align*}
&g_{N}(t, x, y):= - \Delta_x \delta (x-y)
 +(v_N * |\phi|^2 )(t, x) \delta(x-y)\\
&\qquad\qquad  +v_N(x-y) \overline\phi (t, x)  \phi(t, y)
\\
&m_{N}(x, y) :=- v_N(x-y) \phi(x) \phi(y) \ .
\end{align*}
It is clear that $g_{N}$ and $m_{N}$ above depend on the number of particles $N$. Subsequently, for simplicity, we will suppres this
subscript and recall it only when it is relevant in an argument.
Let us define the two operators below
\begin{subequations}
\begin{align}
&\H_{G}:=\frac{1}{2}\int dxdy\left\{-g_{N}(t,x,y)\D_{yx}-g_{N}(t,y,x)\D_{x,y}\right\} \label{diag-g-term}
\\
&\M:=\frac{1}{2}\int dxdy\left\{\overline{m}_{N}(t,x,y)\Q_{xy} +
m_{N}(t,x,y)\Q^{\ast}_{x,y}\right\} \nonumber
\\
&=
 \I\left(\begin{matrix}0&\overline{m}\\ -m&0\end{matrix}\right) \label{off-diag-m-term}
\end{align}
\end{subequations}
so that we can write $\P_{2}=\H_{G}+\M$.
The relevance of this splitting will become clear shortly.
The third and fourth degree polynomiasl arise from the commutators
$
\big[\A,\V\big]
$
and $\quad \V$
respectively
and are given below
\begin{subequations}
\begin{align}
\P_{3}&:=\int dxdy\left\{v_{N}(x-y)\big(\phi(y)a^{\ast}_{x}a^{\ast}_{y}a_{x}+\bar{\phi}(y)a^{\ast}_{x}a_{x}a_{y}\big)\right\}
\label{cubic-1}
\\
\P_{4}&:=(1/2)\int dxdy\left\{v_{N}(x-y)a^{\ast}_{x}a^{\ast}_{y}a_{x}a_{y}\right\}\ . \label{quartic-1}
\end{align}
\end{subequations}
The mean field approximation emerges from the first degree polynomial $\P_{1}$. Since $\mu_{0}$ can be absorbed into the evolution it is reasonable to pick
the field $\phi$ so that $h(\phi)=0$. This leads to the evolution
\begin{equation}
\frac{1}{i}\partial_{t}\phi -\Delta\phi +\big(v_{N}\ast\vert\phi\vert^{2}\big)\phi =0 \label{meanfieldevol}
\end{equation}
which is of Hartree type. The formal limit of the equation above is the cubic NLS where the constant in front of the nonlinear term
is the integral of
the potential $v$. If $\phi$ satisfies \eqref{meanfieldevol} then $\mu_{0}$ reduces to
\begin{equation}
\mu_{0}=\frac{1}{2}\int dxdy\left\{v_{N}(x-y)\vert\phi(t,x)\vert^{2}\vert\phi(t,y)\vert^{2}\right\}\ .
\end{equation}


Now we can compute the reduced Hamiltonian in \eqref{reducedHamiltonian} using the splitting in \eqref{diag-g-term}, \eqref{off-diag-m-term}.
First let us first give a name to
\begin{equation}\label{error-1}
\E:=e^{\B}\left(\P_{3}+N^{-1/2}\P_{4}\right)e^{-B}
\end{equation}
which will be treated later as an error term.
Now we can write, see \eqref{reducedHamiltonian},
\begin{align}\label{reducedHamilt-2}
\H_{red}&=\frac{1}{i}\left(\partial_{t}e^{\B}\right) e^{-\B}
+\H_G +[e^{\B}, \H_G]e^{-\B} + e^{\B} \I(M) e^{-\B} + N \mu_0\nonumber
\\
&-e^{\B}\left(N^{-1/2}\P_{3}+N^{-1}\P_{4}\right)e^{-B}\nonumber +N \mu_0\nonumber
\\
&=\H_G +\I \left((1/i)\left(\partial_{t} e^K\right) e^{-K}
+[e^K, G]e^{-K} + e^K M e^{-K}\right)-N^{-1/2}\E +N \mu_0\nonumber
\\
&=\H_G +\I(R)-N^{-1/2}\E +N \mu_0\ ,
\end{align}
where $R$ is defined to be the expression,
\begin{equation*}
R:=(1/i)\left(\partial_{t} e^K\right) e^{-K}
+[e^K, G]e^{-K} + e^K M e^{-K}\ .
\end{equation*}
For the convenience of the reader, let us recall our set up,
\begin{align*}
&K:=\left(
\begin{matrix}
0 &\overline{k}\\
k & 0
\end{matrix}
\right)
\quad {\rm and}\quad
 e^K=\left(
\begin{matrix}
\ch &\shb\\
\sh & \chb
\end{matrix}
\right)
\\
& \sh:=k + \frac{1}{3!} k \circ \overline k \circ k + \ldots~, \\
& \ch:=\delta(x-y) + \frac{1}{2!}\overline k  \circ k + \ldots~,\\
&g(t, x, y):= - \Delta_x \delta (x-y)
 +(v_N * |\phi|^2 )(t, x) \delta(x-y)\\
&\qquad\qquad +v_N(x-y) \overline\phi (t, x)  \phi(t, y)
~, \\
&m(x, y) :=- v_N(x-y) \phi(x) \phi(y)
\quad {\rm where}\quad
v_N(x)= N^{3 \beta} v(N^{\beta}x)\\
&G:=
\left(
\begin{matrix}
g &0\\
0 & -g^T
\end{matrix}
\right)\quad {\rm and}\quad M :=
\left(
\begin{matrix}
0 &\overline{m}\\
-m & 0
\end{matrix}
\right)\\
&N_u :=
\left(
\begin{matrix}
I &0\\
0 & -I
\end{matrix}
\right) \, \mbox{ (this corresponds to the Number operator)}\\
&{\bf S}(s):= \frac{1}{i} s_t  + g^T \circ s + s \circ g\quad {\rm and}\quad
{\bf W}(p):= \frac{1}{i} p_t  +[g^T, p]\ .
\end{align*}
Thus ${\bf S}$ describes a Shr\"odinger type evolution, while ${\bf W}$ is a Wigner type operator. These operators will emerge shortly.
Recall the formula \eqref{reducedHamilt-2} that we derived earlier for the reduced Hamiltonian
\begin{align*}
\H_{red}=\H_G +\I(R)-N^{-1/2}\E\ ,
\end{align*}
where $\H_G$ has only $a^* a$  terms (which annihilate the vacuum)  and $R$ can be computed explicitly. In fact we have,
\begin{align*}
&R=\\
&\bigg(
\frac{1}{i}
 \left(
\begin{matrix}
\ch_t & \shb_t\\
\sh_t & \chb_t
\end{matrix}
\right)\\
& +
\left(
\begin{matrix}
[\ch, g] -\shhb \circ  m& \, \, \, -\shb\circ g^T -g\circ \shb+ \chh\circ \overline m\\
\sh\circ g + g^T\circ \sh -\chhb\circ  m& -[\chb, g^T]+\shh\circ \overline m
\end{matrix}
\right)\bigg)\\
&\circ\left(
\begin{matrix}
\ch &-\shb\\
-\sh & \chb
\end{matrix}
\right) \, \,\\
 &\text{(matrix product, where kernel products mean compositions)}
\end{align*}

The condition that we would like to impose is that $R$ is block diagonal so that $\I(R)$ contains only terms of the form
$aa^{\ast}$ and $a^{\ast}a$ so that, apart from a trace when we commute $a$ with $a^{\ast}$, we obtain an operator which
annihilates the vacuum state. The remaining trace can be absorbed in the evolution as a phase factor. Thus our requirement is
\begin{align}
\frac{1}{i}\left(\frac{\partial}{\partial t} e^K\right) e^{-K}
+[e^K, G]e^{-K} + e^K M e^{-K} \, \, \, \, \mbox {is block diagonal.}\label{cond-a}
\end{align}
We proceed to show this equivalent to equations \eqref{pair-S-1}, \eqref{pair-W-1}.
Let us make the elementary observations
\begin{align*}
&\left(\frac{\partial}{\partial t} e^K\right) e^{-K}=
\frac{\partial}{\partial t} I - e^K \circ \frac{\partial}{\partial t} \circ e^{-K}\\
&\ \, [e^K, G]e^{-K} =e^K G e^{-K} -G
\end{align*}
so removing the part of \eqref{cond-a} that is diagonal already we have the equivalent formulation of \eqref{cond-a}
\begin{align}
e^K\left(-\frac{1}{i}\frac{\partial}{\partial t} +G +M\right) e^{-K} \, \, \, \, \mbox {is block diagonal.}\label{cond-b}
\end{align}
Now we make the  observation that a matrix is block-diagonal if and only if it commutes with the number operator matrix $N_u$, as well as
(for arbitrary matrices $A$ and $B$) we have $[e^K A e^{-K}, B]=0$ if and only if $[A, e^{-K} B e^K]=0$, so our equation \eqref{cond-b} reads,
\begin{align}
\left[\left(-\frac{1}{i}\frac{\partial}{\partial t} +G +M\right)\, ,\, e^{-K} N_u e^{K}\right]=0\ .\label{cond-c}
\end{align}
A direct calculation gives
\begin{align*}
e^{-K} N_u e^{K}=
\left(
\begin{matrix}
\cht &\shbt\\
-\sht & -\chbt
\end{matrix}
\right)
\end{align*}
after which is is straightforward to compute
\begin{align*}
\left[-\frac{1}{i}\frac{\partial}{\partial t} +   G\, ,\, e^{-K} N_u e^{K}\right]=
\left(
\begin{matrix}
\overline{{\bf W}\left(\chbt \right)}&\overline{{\bf S}\left((\sht\right)}\\
{\bf S}\left(\sht\right) & {\bf W}\left(\chbt\right)
\end{matrix}
\right)
\end{align*}
and simlarly,
\begin{align*}
\left[M\, ,\, e^{-K} N_u e^{K}\right]=
\left(
\begin{matrix}
\overline{-m\circ \shbt + \sht \circ \overline{m}} &\overline{- m \circ \cht - \chbt \circ m}\\
- m \circ \cht - \chbt \circ m & -m\circ \shbt + \sht \circ \overline{m}
\end{matrix}
\right)
\end{align*}
Finally combining the two formulas above we obtain, see  \eqref{cond-c}, the linear pair of equations below
\begin{subequations}
\begin{align}
&{\bf S}\left(\sht\right)=m \circ \cht + \chbt \circ m \label{pair-S-2}\\
&{\bf W}\left(\chbt\right)= m \circ \shbt- \sht \circ \overline m\ .\label{pair-W-2}
\end{align}
\end{subequations}
This completes the derivation of the evolution equations for the pair excitations and the mean field, namely \eqref{pair-S-2}, \eqref{pair-W-2},
together with \eqref{meanfieldevol}  describe the evolution
of $\phi$ and $k$ and are the equations in \eqref{meanfield-3}, \eqref{pair-S-1} and \eqref{pair-W-1}.
In particular, we have proved that in that if $\phi$, $k$ satisfy these equations, then the energy estimate \eqref{energyestim1} holds.

\section{Estimates for the solution to the Hartree equation}
This section adapts classical results for NLS due to Lin and Strauss \cite{L-S}, Ginibre and Velo \cite{G-V-nls}, Bourgain \cite{B-nls}, as well as
Colliander, Keel, Staffilani, Takaoka and Tao \cite{CKSTT} to the Hartree equation.
Assume
\begin{align}
&\frac{1}{i} \frac{\partial}{\partial t} \phi - \Delta \phi + \left(v_N * |\phi|^2 \right)\phi =0 \label{har}\\
&\phi(0, \cdot)=\phi_0 \ . \notag
\end{align}
where
$v \in C^1_0$  is non-negative, spherically symmetric, and decreasing.
We recall the relevant conserved quantities, following the notation \cite{GMM2}:
\begin{align*}
&\rho :=(1/2)|\phi|^2~; \\
&p_{j}:=(1/2i)\left(\phi\nabla_{j}\overline \phi -\overline \phi\nabla_{j}\phi\right)~;
\quad p_{0}=(1/2i)\left(\phi\partial_{t}\overline{\phi } -\overline{\phi}\partial_{t}\phi\right)~; \\
&\sigma_{jk}:=\nabla_{j}\overline{\phi} \nabla_{k}\phi +\nabla_{k}\overline{\phi}\nabla_{j}\phi~;
\quad \sigma_{0j}=\nabla_{j}\overline{\phi}\partial_{t}\phi +\partial_{t}\overline{\phi}\nabla_{j}\phi~\\
&\lambda :=-\Im (\phi \partial_t \overline{\phi})+ |\nabla \phi|^2 +
\frac{1}{2}(v * |\phi|^2) |\phi|^2\\
&=\frac{1}{2}\left(\Delta |\phi|^2 - (v * |\phi|^2) |\phi |^2 \right);\\
e:&= |\nabla \phi|^2 + \frac{1}{2}(v * |\phi|^2) |\phi |^2~.
\end{align*}

The associated conservation laws are
\begin{subequations}
\label{conserv-laws}
\begin{align}
&\partial_{t}\rho -\nabla_{j}p^{j}=0~,\label{conserv-a} \\
&\partial_{t}p_{j} -\nabla_{k}\left\{\sigma_{j}^{\ k}-\delta_{j}^{\ k}\lambda\right\} +l_{j}=0~, \label{conserv-b}\\
&\partial_{t}e -\nabla_{j}\sigma_{0}^{\ j}+l_{0} =0~.\label{conserv-c}
\end{align}
\end{subequations}
These laws express the conservation of mass, momentum and energy, respectively,
where the vector $\big(l_{j}, l_{0}\big)$ is
\begin{align*}
l_{j}:=2 \left((v_N * \rho) \rho_j - (v_N* \rho_j) \rho\right)\ ,\qquad l_{0}:=2 \left((v_N * \rho) \rho_0 - (v_N* \rho_0) \rho\right)\ .
\end{align*}
In the case of NLS, $v_N=\delta$ and $l_j, l_0$ are 0.

We adapt the well-known method of interaction Morawetz estimates, due to Colliander, Keel, Staffilani, Takaoka and Tao,  outlined in \cite{CKSTT}.
Start with
\begin{align*}
Q(t)= \int \left(\nabla_j p^j(t, x) \rho(t, y) + \rho(t, x)\nabla_j p^j(t, x) \right)|x-y| dx dy\ .
\end{align*}
Using \eqref{conserv-laws} we get
\begin{align*}
&\dot{Q}(t)= 2 \int \nabla_j p^j(t, x) \nabla_k p^k(t, y)|x-y| dx dy\\
&+\int \bigg(\nabla_j \left(\nabla_{k}\left\{\sigma_{j}^{\ k}(t, x)-\delta_{j}^{\ k}\lambda(t, x)\right\} -l_{j}(t, x) \right) \rho(t, y) \\
 &+ \rho(t, x)\nabla_j \left(  \left(\nabla_{k}\left\{\sigma_{j}^{\ k}(t, y)-\delta_{j}^{\ k}\lambda(t, y)\right\} -l_{j}(t, y) \right)  \right)\bigg)|x-y| dx dy\\
 &\ge \int \left(-\lambda(t, x)  \rho(t, y) - \rho(t, x) \lambda(t, y)\right)\Delta |x-y| dx dy \, \,  \mbox {(main term)}\\
 &- \int \left((\nabla_j l_{j}(t, x) )\rho(t, y) + \rho(t, x)(\nabla_j l_{j}(t, y)) \right) |x-y| dx dy \, \,  \mbox {(error term.)}
\end{align*}
We have used the fact which we recall for the reader's convenience (see \cite{CKSTT}), that
\begin{align*}
&\left(\nabla_j \nabla_k  a\right)(x-y)\left( -2 p_j (t, x) p_k(t, y) + \sigma_{j}^{\ k}(t, x) \rho(t, y) +
\sigma_{j}^{\ k}(t, y) \rho(t, x)\right)\\
&= \left(\nabla_j \nabla_k  a\right)(x-y)\bigg(\left(\phi(x) \overline {\phi_j}(y) + \phi_j(x) \overline \phi(y)\right)
\overline{\left(\phi(x) \overline {\phi_j}(y) + \phi_j(x) \overline \phi(y)\right)}\\
&+\left(\phi(x) \phi_j(y) - \phi_j(x)  \phi(y)\right)\overline{\left(\phi(x) \phi_j(y) - \phi_j(x)  \phi(y)\right)}\bigg)
\ge 0
\end{align*}
where $a(x)=|x|$.
It is easy to check
\begin{align*}
&\mbox {(main term)} \ge c \|\phi(t, \cdot)\|^4_{L^4} \\
&+ 2 \int\big( (v_N*\rho)(t, x) \rho(t, x) \rho(t, y) +
(v_N*\rho)(t, y) \rho(t, y) \rho(t, x)\big) \Delta|x-y| dx dy\\
& \mbox{with $c>0$}
\end{align*}
 We proceed to analyze the error term:
\begin{align*}
&\mbox{error term}\\
&=- \int \left((\nabla_j l_{j}(t, x) )\rho(t, y) + \rho(t, x)(\nabla_j l_{j}(t, y)) \right) |x-y| dx dy  \\
&=-2\int (\nabla_j l_{j}(t, x) )\rho(t, y)  |x-y| dx dy\\
&=2\int  l_{j}(t, x) \rho(t, y)  \frac{(x-y)^j}{|x-y|} dx dy\\
&=4 \int  \big((v_N * \rho)(t, x) \rho_j(t, x) - (v_N* \rho_j)(t, x) \rho(t, x)\big)\rho(t, y)  \frac{(x-y)^j}{|x-y|} dx dy\\
&=4 \int v_N(x-z)(\big( \rho(t, z) \rho_j(t, x)- \rho_j(t, z) \rho(t, x)\big)\rho(t, y)\frac{(x-y)^j}{|x-y|} dx dy dz\\
&=-8\int v'_{N}(|x-z|) \frac{(x-z)^j}{|x-z|}(\rho(t, z) \rho(t, x)\rho(t, y)\frac{(x-y)^j}{|x-y|} dx dy dz\\
&- 4 \int v_N(x-z)\rho(t, z) \rho(t, x)\rho(t, y) \partial_{x, j} \left(\frac{(x-y)^j}{|x-y|}\right) dx dy dz\\
&=-4 \int v'_N(|x-z|) \left( \frac{(x-z)^j}{|x-z|}\frac{(x-y)^j}{|x-y|}  +\frac{(z-x)^j}{|z-x|}\frac{(z-y)^j}{|z-y|}\right)  \rho(t, z) \rho(t, x)\rho(t, y)\\
&- 2 \int\big( (v_N*\rho)(t, x) \rho(t, x) \rho(t, y) +
(v_N*\rho)(t, y) \rho(t, y) \rho(t, x)\big) \Delta|x-y| dx dy
\end{align*}
The next-to-last line is $\ge 0$ because of the assumption  $v'_N \le 0$ and the elementary trigonometric inequality
\begin{align*}
&\frac{(x-z)^j}{|x-z|}\frac{(x-y)^j}{|x-y|} + \frac{(z-x)^j}{|z-x|}\frac{(z-y)^j}{|z-y|} \\
&= \cos(\theta_1) + \cos(\theta_2) \ge 0
\end{align*}
 The last line is negative, but cancels part of the main term. Thus
\begin{align*}
\mbox{(main term) + (error term)} \ge c \|\phi\|^4_{L^4}
\end{align*}
Since $Q(t)$ is bounded uniformly in time by $\|\phi_0\|^4_{H^1}$,
we have shown the following proposition.
\begin{proposition} Let $\phi$ be a solution to the Hartree equation \eqref{har}. There exists $C$ depending only on $\|\phi_0\|_{H^1}$ such that
\begin{align*}
\|\phi\|_{L^4([0, \infty) \times \mathbb R^3)} \le C
\end{align*}
and, as an immediate consequence of conservation of energy,
\begin{align}
\|\phi\|_{L^8[0, \infty) L^4( \mathbb R^3)} \le C\ .  \label{harest}
\end{align}
\end{proposition}

\begin{remark}
It was shown by Bourgain in \cite{B-nls} that if $\phi$ is a solution to cubic NLS, then there exists $C_s$ depending only on
$\|\phi_0\|_{H^s}$ such that
\begin{align*}
\|\phi(t, \cdot)\|_{H^s} \le C_s  \, \, \forall t\ .
\end{align*}
\end{remark}

Using the above Morawetz estimate (which was not yet discovered when Bourgain did this work), we can easily prove
the same for our Hartree equation.
\begin{proposition} \label{Hs} Let $\phi$ be a solution to the  equation \eqref{har}. There exists $C_s$
depending only on
$\|\phi_0\|_{H^s}$ such that
such that
\begin{align*}
\|\phi(t, \cdot)\|_{H^s} \le C_s
\end{align*}
uniformly in time.
\end{proposition}
\begin{proof}
Split $[0, \infty)$ into finitely many intervals $I_k$ where
\begin{align*}
\|\phi\|_{L^8(I_k) L^4( \mathbb R^3)} \le \epsilon
\end{align*}
where $\epsilon$ is to be prescribed later.
Differentiating \eqref{har}
\begin{align}
&\frac{1}{i} \frac{\partial}{\partial t} D^s\phi - \Delta D^s\phi =- D_s \left((v_N* |\phi|^2) \phi\right) \label{dhar}\\
& \mbox{where} \notag\\
& D_s \left((v_N* |\phi|^2) \phi\right) =\left(v_N * |\phi|^2 \right)D^s\phi \,
\, \mbox{ $+$ similar and easier terms.} \notag
\end{align}
For the first interval, $I_1$, we get, using the $L^{8/3}L^{4}$ Strichartz estimate,
\begin{align*}
&\|D^s \phi\|_{L^{8/3}(I_1)L^{4}(\mathbb R^3)} \le C \|\phi_0\|_{H^s} + C\| \left(v_N * |\phi|^2 \right)D^s\phi\|_{L^{8/5(I_1)}L^{4/3}(\mathbb R^3)}\\
& \le C_1 \|\phi_0\|_{H^s} + C_2 \|\phi\|^2_{L^8(I_1) L^4(\mathbb R^3)} \|D^s \phi\|_{L^{8/3}(I_1)L^{4}(\mathbb R^3)}\ .
\end{align*}
At this stage, we pick $\epsilon$ so that $C_2 \epsilon^2 \le \frac{1}{2}$ to conclude
\begin{align*}
\|D^s \phi\|_{L^{8/3}(I_1)L^{4}(\mathbb R^3)}
 \le 2 C_1 \|\phi_0\|_{H^s}\ .
\end{align*}
In turn, this allows us to control the inhomogeneity of \eqref{dhar}
\begin{align*}
&\| \left(v_N * |\phi|^2 \right)D^s\phi\|_{L^{8/5}(I_1)L^{4/3}(\mathbb R^3)}\\
& \le \|\phi\|^2_{L^8(I_1) L^4(\mathbb R^3)} \|D^s \phi\|_{L^{8/3}(I_1)L^{4}(\mathbb R^3)}\\
&\le C  \|\phi_0\|^3_{H^s}
\end{align*}
and therefore
\begin{align*}
&\|\phi(t, \cdot)\|_{H^s} \le \|\phi_0\|_{H^s} + C \|\phi\|^2_{L^8(I_1) L^4(\mathbb R^3)} \|D^s \phi\|_{L^{8/3}(I_1)L^{4}(\mathbb R^3)}\\
&\le C  \|\phi_0\|^3_{H^s}
\end{align*}
for all $t \in I_1$.
Repeating the process finitely many times, we are done.
\end{proof}

If we assume the data $\phi_0$ and sufficiently many derivatives  are not only in $L^2$ but also in $L^1$, we can also get decay.
\begin{corollary} \label{linf}Let $\phi$ be a solution to \eqref{har}. There exists $C$ depending only on $\|\phi_0\|_{W^{k, 1}}$
for $k$ sufficiently large such that
\begin{subequations}
\begin{align}
&\|\phi(t, \cdot)\|_{L^{\infty}} \le \frac{C}{t^{\frac{3}{2}}} \label{linfest}\\
&\mbox{and also} \notag\\
&\|\partial_t \phi(t, \cdot)\|_{L^{\infty}} \le \frac{C}{t^{\frac{3}{2}}}
\end{align}
\end{subequations}
\end{corollary}
\begin{proof}
The proof follows the outline of \cite{L-S}, except that
we have two modern ingredients which were not available to Lin and Strauss in 1977:
\begin{align*}
&\|\phi\|_{C^s(\mathbb R^{3+1})} \le C_s \, \, s \in \mathbb N\\
&\|\phi\|_{L^4(\mathbb R^{3+1})} \le C
\end{align*}
This implies that $\|\phi(t, \cdot)\|_{L^{\infty}(\mathbb R^{3})} \to 0 $ as $t \to \infty $. Indeed,
\begin{align*}
&\|\nabla (\phi^2)\|_{L^4([n, n+1] \times \mathbb R^3)} \\
&\le 2 \|\nabla \phi\|_{L^{\infty}(\mathbb R^{3+1})}
\|\phi\|_{L^4([n, n+1] \times \mathbb R^3)} \to 0
\end{align*}
This implies $\|\phi\|_{L^p([n, n+1] \times \mathbb R^3)} \to 0$ for any fixed $4<p<\infty$. Repeating the process one more time implies  $\|\phi(t, \cdot)\|_{L^{\infty}(\mathbb R^{3})} \to 0 $.

We solve \eqref{har} by Duhamel's formula and use the standard $L^{\infty}$ $L^1$ decay estimate for the linear equation. We use the following  estimate:
\begin{align}
&\|e^{i (t-s) \Delta} \left(( v*|\phi|^2 )\phi (s)\right)\|_{L^{\infty}} \label{ex1}\\
& \le \notag
 \frac{C}{|t-s|^{3/2}}
 \|( v*|\phi|^2) \phi (s)\|_{L^1} \le \frac{C}{|t-s|^{3/2}} \|\phi(s, \cdot)\|_{L^{\infty}}
\end{align}
We would also like to estimate $\|e^{i (t-s) \Delta} \left(( v*|\phi|^2 )\phi (s)\right)\|_{L^{\infty}}$ \\
by $\|\nabla e^{i (t-s) \Delta} \left(( v*|\phi|^2 )\phi (s)\right)\|_{L^3}$. This is a false end-point, but becomes true if one replaces $3$ by $3+ \epsilon$. To keep numbers easy, skip the $\epsilon$ and  notice first that
\begin{align}
&\|\nabla e^{i (t-s) \Delta} \left(( v*|\phi|^2 )\phi (s)\right)\|_{L^3} \label{ex3}
 \le \frac{C}{|t-s|^{1/2}}
\|\nabla( v*|\phi|^2 )\phi (s) \|_{L^{3/2}}\\
&\le \frac{C}{|t-s|^{1/2}} \|\nabla \phi\|_{L^2}\|\phi^2\|_{L^6}\notag
\le \frac{C}{|t-s|^{1/2}} \|\phi(s, \cdot)\|^{2/3}_{L^4}\|\phi(s, \cdot)\|^{4/3}_{L^{\infty}}\\
&\le \frac{C}{|t-s|^{1/2}} \|\phi(s, \cdot)\|^{4/3}_{L^{\infty}} \notag
\end{align}
Now, using $3+ \epsilon$ rather than $3$ leads to an estimate of the form
\begin{align}
&\|e^{i (t-s) \Delta} \left(( v*|\phi|^2 )\phi (s)\right)\|_{L^{\infty}} \label{ex2}
\le \frac{C}{|t-s|^{1/2+ \epsilon'}} \|\phi(s, \cdot)\|^{4/3 - \epsilon''}_{L^{\infty}}
\end{align}
Combining \eqref{ex1} and \eqref{ex2} we get: There exists a kernel $ \in L^1([0, \infty))$ and $\delta>0$ such that
\begin{align}
&\|e^{i (t-s) \Delta} \left(( v*|\phi|^2 )\phi (s)\right)\|_{L^{\infty}} \label{ex3}
\le k(t-s) \|\phi(s, \cdot)\|^{1+\delta}_{L^{\infty}}
\end{align}
Putting all together
\begin{align*}
\|\phi(t, \cdot)\|_{L^{\infty}} \le \frac{C}{t^{3/2}} \|\phi_0\|_{L^1} + \int_0^{t/2} \frac{C}{|t-s|^{3/2} }\|\phi(s, \cdot)\|_{L^{\infty}} ds +\int_{t/2}^t k(t-s)\|\phi(s, \cdot)\|^{1+\delta}_{L^{\infty}} ds
\end{align*}
Denoting $M(t)= \sup_{0 <s<t} (1+s^{3/2}) \| \phi(s, \cdot)\|_{L^{\infty}}$,
We have, for $t>1$,
\begin{align*}
M(t) \le C \|\phi_0\|_{L^1} + \frac{C}{(1+|t|^{3/2})}\int_0^{t/2} \frac{1}{(1+|s|^{3/2})} M(s) ds
+ C \sup_{t/2<s<t}\|u(s, \cdot)\|^{\delta}_{L^{\infty}} M(t)
\end{align*}
The last term can be absorbed in $M(t)$, and the result follows by Gronwall's inequality.
Now that we know that $\|\phi(s, \cdot)\|_{L^{\infty}} \le \frac{C}{1+s^{3/2}}$, it is very easy to estimate $\partial_t \phi$. We use \eqref{ex1} and \eqref{ex3} (with $3+\epsilon$ replacing 3), as well as the fact that all norms
$\|\partial^{\alpha} \phi(s, \cdot)\|_{L^p} \le C_{\alpha, p}$ uniformly in $s$, for all $\p \ge 2$. This is a consequence of
Proposition \eqref{Hs} (boundedness of the $H^s$ norms).
\begin{align*}
&\|\partial_t \phi(t, \cdot)\|_{L^{\infty}} \le \frac{C}{t^{3/2}} \|\partial_t\phi_0\|_{L^1} + \int_0^t
\|e^{i (t-s) \Delta} \partial_s \left(( v*|\phi|^2 )\phi (s)\right)\|_{L^{\infty}}\\
& \le \frac{C}{t^{3/2}}+ C \int_0^{t-1} \frac{1}{1+|t-s|^{3/2}}\|\partial_s \left(( v*|\phi|^2 )\phi (s)\right)\|_{L^1} ds\\
&+ C \int_{t-1}^{t} \frac{1}{1+|t-s|^{1/2 + \epsilon}}\|\nabla \partial_s \left(( v*|\phi|^2 )\phi (s)\right)\|_{L^{3/2-\epsilon'}} ds\\
& \le \frac{C}{t^{3/2}}+ C \int_0^{t-1} \frac{1}{1+|t-s|^{3/2}}\|\partial_s \left(( v*|\phi|^2 )\phi (s)\right)\|_{L^1} ds\\
&  + C \int_{t-1}^{t} \frac{1}{1+|t-s|^{1/2 + \epsilon}}\|\nabla \partial_s \left(( v*|\phi|^2 )\phi (s)\right)\|_{L^{3/2-\epsilon'}} ds\\
&  \le \frac{C}{t^{3/2}} + C \int_0^{t-1} \frac{1}{1+|t-s|^{3/2}}\|\phi (s)\|_{L^{\infty}} ds\\
&+ C \int_{t-1}^{t} \frac{1}{1+|t-s|^{1/2 + \epsilon}}\|\phi (s)\|_{L^{\infty}} ds\\
\end{align*}
If we estimate $\|\phi (s)\|_{L^{\infty}}$ using \eqref{linfest}, we are done.

\end{proof}

By interpolating with the $L^2$ uniform bound we get the next Corollary.
\begin{corollary}
\label{l3}Let $\phi$ be a solution to \eqref{har}. There exists $C$ depending only on $\|\phi_0\|_{W^{k, 1}}$
for $k$ sufficiently large such that
\begin{align*}
\|\phi(t, \cdot)\|_{L^3} + \|\partial_t \phi(t, \cdot)\|_{L^3}  \le \frac{C}{1+t^{\frac{1}{2}} }\ . \\
\end{align*}
\end{corollary}

\section{Estimates for the pair excitations}

Define $\cht := \delta + p_2$, $\sht :=s_2$, and also
$\ch := \delta + p_1$, $\sht :=s_1$
so that, see \eqref{pair-S-2}, \eqref{pair-W-2} become
\begin{subequations}
\begin{align}
&{\bf S}\left(s_2\right)= 2 m +m \circ p_2 + \pb_2 \circ m \label{Ppair}\\
&{\bf W}\left(\bar{p}_2\right)= m \circ \sb_2- s_2 \circ \overline m \label{oldPpair}\\
&s_2(0, \cdot)=p_2(0, \cdot)=0 \notag
\end{align}
\end{subequations}
The goal of this section is to prove the following theorem.
\begin{theorem} \label{sc4thm} Assume $\phi_0 \in  W^{k, 1}$ for $k$  sufficiently large.
The following estimates hold:
\begin{align}
\|s_2(t, \cdot)\|_{L^2(\mathbb R^6)} + \|p_2(t, \cdot)\|_{L^2(\mathbb R^6)} \label{not}
\le C\log ( 1+ t)
\end{align}
where $C$  depends on  $\|\phi(0, \cdot)\|_{ W^{k, 1}}$ for some finite $k$. A similar result holds for the higher time derivatives, but we will not use it or prove it.
\end{theorem}
An immediate corollary is of the above theorem is,
\begin{corollary} \label{s1est}
The following estimates hold:
\begin{align*}
&\|s_1(t, \cdot)\|_{L^2(\mathbb R^6)}
\le C\log ( 1+ t)\\
&\|p_1(t, \cdot)\|_{L^2(\mathbb R^6)}
\le C\log ( 1+ t)\\
&\int |p_1(x, x)| dx \le  C\log^2 ( 1+ t)\ .
\end{align*}
\end{corollary}
\begin{proof} (of corollary \eqref{s1est})
Since $\sht = 2 \sh \circ \ch$, we get
\begin{align*}
&\|s_1(t, \cdot)\|_{L^2(\mathbb R^6)} \le \frac{1}{2} \|s_2(t, \cdot)\|_{L^2(\mathbb R^6)} \|\ch^{-1}\|_{operator}\\
&\le  \frac{1}{2} \|s_2(t, \cdot)\|_{L^2(\mathbb R^6)}\ .
\end{align*}
We also have  $p_1(x, x) \ge 0$, $p_1\circ p_1(x, x) \ge 0$, so taking traces in the relation
\begin{align*}
p_1\circ p_1 + 2 p_1 = \overline{s}_1 \circ s_1
\end{align*}
gives the other estimates.
\end{proof}

Before starting the proof of Theorem \eqref{sc4thm} , we need some preliminary lemmas.
\begin{lemma} \label{ker}
Recall  $
m(t, x, y) =- v_N(x-y) \phi(t, x) \phi(t, y)$. Then there exists $C$ such that
\begin{subequations}
\begin{align}
&\int \frac{|\widehat{ m_N} (t, \xi, \eta)|^2}{\left(|\xi|^2 + |\eta|^2\right)^2} d \xi d \eta\le C \label{mint}
 \|\phi(t, \cdot)\|_{L^3}^4\\
 &\mbox{and also} \notag\\
 &\int \frac{|\partial_t \widehat{ m_N} (t, \xi, \eta)|^2}{\left(|\xi|^2 + |\eta|^2\right)^2} d \xi d \eta
 \le  \|\phi(t, \cdot)\|_{L^3}^2 \| \partial_t\phi(t, \cdot)\|_{L^3}^2\ . \label{mint1}\\
 \notag
\end{align}
\end{subequations}
Similar estimates hold for higher time derivatives.
\end{lemma}
\begin{proof}
Write
\begin{align*}
v_N(x-y)\phi(t, x) \phi(t, y) = \int \delta(x-y-z)v_N(z) \phi(t, x) \phi(t, y) dz\ .
\end{align*}
The Fourier transform of $\delta(x-y-z) \phi(t, x) \phi(t, y)$ is easily computed to be
\begin{align*}
e^{i z \cdot \eta} \widehat{\phi \phi_z}(t, \xi + \eta)
\end{align*}
where we denote $\phi_z(x)=\phi(x-z)$.
Thus
\begin{align*}
&|\widehat{ m_N} (t, \xi, \eta)|^2=|\int v_N(z) e^{i z \cdot \eta} \widehat{\phi \phi_z}(t, \xi + \eta) dz|^2\\
&\le \|v_N\|_{L^1} \int |v_N(z)|| \widehat{\phi \phi_z}(t, \xi + \eta)|^2 dz
\end{align*}
Thus, after a change of variables, the left hand side of \eqref{mint} is dominated by
\begin{align*}
&\int |v_N(z)| \frac{|\widehat{\phi \phi_z}(t, \xi )|^2}{\left(|\xi|^2 + |\eta|^2\right)^2} d \xi d \eta dz\\
&\le C \int |v_N(z)| \frac{|\widehat{\phi \phi_z}(t, \xi )|^2}{|\xi|} d \xi  dz\\
&\le C \int |v_N(z)|\|\phi\|^4_{L^3} dz =C\|\phi\|^4_{L^3}\ .
\end{align*}
We have used the fact that
\begin{align*}
&\int\frac{|\widehat{\phi \phi_z}(t, \xi )|^2}{|\xi|} d \xi \le C \|D_x^{-1/2}\left(\phi \phi_z\right)\|^2_{L^2}\\
&\le C \|\phi \phi_z\|^2_{L^{3/2}} \, \, \mbox{(by Hardy-Littlewood-Sobolev)}\\
&\le C \|\phi\|^4_{L^3}\ .
\end{align*}
The proof of \eqref{mint1} is similar.
\end{proof}
\begin{lemma} \label{s_a}
Let $s_a^0$ be the solution to
\begin{align}
&\left(\frac{1}{i} \frac{\partial}{\partial t} -\Delta_{\mathbb R^6} \right)s^0_a(t, x, y) = 2 m(t, x, y) \label{Seq1}\\
&s^0_a(0, x, y)=0\ . \notag
\end{align}
Then
\begin{align*}
&\|s^0_a(t, \cdot)\|_{L^2(\mathbb R^6)}\\
 &\le C \bigg(\|\phi(0, \cdot)\|_{L^3 }^2 + \|\phi(t, \cdot)\|_{L^3}^2
+
&  \int_0^t  \|\phi(s, \cdot)\|_{L^3 }
\|\partial _s\phi(s, \cdot)\|_{L^3 } ds\bigg)\\
&\le C \log (1+t)
\end{align*}
where $C$ depends only on  $\|\phi_0\|_{W^{k, 1}}$.
\end{lemma}
\begin{proof}
Solving \eqref{Seq1} by Duhamel's formula we get
\begin{align}
&\|s^0_a(t, \cdot)\|_{L^2(\mathbb R^6)} \notag\\
&\le C
\|\int_0^t e^{i s \left(|\xi|^2 + |\eta|^2\right)} \hat{m}(s, \xi, \eta) d s \|_{L^2(\mathbb R^6)}\notag \\
&\le C \bigg(
\|\frac{ \hat{m}(0, \xi, \eta)}{|\xi|^2 + |\eta|^2 }  \|_{L^2(\mathbb R^6)}
+\|\frac{ \hat{m}(t, \xi, \eta)}{|\xi|^2 + |\eta|^2 }  \|_{L^2(\mathbb R^6)} \bigg) \label{terms}\\
&+\|\int_0^t e^{i s \left(|\xi|^2 + |\eta|^2\right)}\frac{ \frac{\partial}{\partial s}\hat{m}(s, \xi, \eta)}{|\xi|^2 + |\eta|^2 } d s \|_{L^2(\mathbb R^6)}
\ .\label{terms1}
\end{align}
The terms \eqref{terms} are estimated directly by Lemma \eqref{ker}.
For the last term \eqref{terms1} move the norm inside the integral and use Lemma \eqref{ker} as well as Corollary
\eqref{l3}.
\end{proof}

\begin{lemma} \label{lem3}
Let $s_a$ be the solution to
\begin{align}
&\S(s_a) = 2 m(t, x, y) \label{eq1}\\
&s_a(0, x, y)=0\ . \notag
\end{align}
Then
\begin{align*}
 \|s_a(t, \cdot)\|_{L^2(\mathbb R^6)} \le C \log(1+t)
 \end{align*}
where $C$  depends  on  $ \|\phi_0\|_{W{k, 1}}$.
\end{lemma}
\begin{proof}
Let $V$ be the "potential" part of $\S$, so that
\begin{align*}
\S= \frac{1}{i} \frac{\partial}{\partial t} -\Delta_{\mathbb R^6} + V\ .
\end{align*}
 The operator $V$  is  bounded from $L^2$ to $L^2$, with norm bounded by
$\|\phi(t, \cdot)\|^2_{L^{\infty}} \le \frac{C}{1+ t^3}$ (by Corollary \eqref{linf}).
Explicitly,
\begin{align*}
&V(u)(t, x, y)=  \left( (v_N * |\phi|^2 )(t, x) +  (v_N * |\phi|^2 )(t, y) \right) u(t, x, y)\\
&+
\int v_N(x-z) \overline\phi (t, x)  \phi(t, z) u(z, y) dz
+ \int u (x, z)v_N(z-y) \phi (t, z)  \overline \phi(t, y)  dz\ .
\end{align*}
Write
\begin{align*}
s_a= s_a^0 + s_a^1
\end{align*}
where $s_a^0$ is as in the previous lemma,
\begin{align*}
\frac {1}{i} \frac{\partial}{\partial t}s^0_a(t, x, y) -\Delta_{\mathbb R^6} s^0_a(t, x, y)  = 2 m\ .
\end{align*}
Then $s_a^1$ satisfies the equation
\begin{align}
\S(s_a^1)  =- V \left( s_a^0(t, \cdot)\right)\ . \label{ener}
\end{align}
Both $s_a^0$ and $s_a^1$ are zero initially, and the estimate is already known for
$s_a^0$.
. We apply energy estimates to \eqref{ener}.
Recall
\begin{align*}
 \W (s_a^1 \circ \overline {s_a^1})=\S(s_a^1) \circ \overline {s_a^1} - s_a^1 \circ \overline {\S(s_a^1)} =  - V \left( s_a^0(t, \cdot)\right) \circ \overline {s_a^1} + s_a^1 \circ \overline {V \left( s_a^0(t, \cdot)\right)}\ .
\end{align*}
Taking traces, we get
\begin{align*}
&\frac{\partial}{\partial t} \|s_a^1\|^2_{L^2(\mathbb R^6)} \le 2 \|V \left( s_a^0(t, \cdot)\right)\|_{L^2(\mathbb R^6)}
\|s_a^1\|_{L^2(\mathbb R^6)}\\
&\le \frac{C}{1+t^3} \| s_a^0(t, \cdot)\|_{L^2(\mathbb R^6)}
\|s_a^1\|_{L^2(\mathbb R^6)}\\
&\le \frac{ C \log (1+t)}{1+t^3} \|s_a^1(t, \cdot)\|_{L^2(\mathbb R^6)}\ .
\end{align*}
Integrating, we get the estimate
\begin{align*}
 \|s_a^1(t, \cdot)\|_{L^2(\mathbb R^6)} \le C
 \end{align*}
 which implies the claim.
\end{proof}
We are ready for the proof of Theorem \eqref{sc4thm}.
\begin{proof}
Write $s_2=s_a+s_e$ where ${\bf S}\left(s_a\right)=2m$, as in the previous lemma. The kernels
$s_e$ and $p_2$ satisfy the following less singular system:
\begin{subequations}
\begin{align}
&{\bf S}\left(s_e\right)= m \circ p_2 + \bar{p}_2 \circ m \label{newpair1}\\
&{\bf W}\left(\pb_2\right)= \big(m \circ \sab- s_a \circ \overline m \big) +  m \circ \seb- s_e \circ \overline m\ . \label{newpair2}
\end{align}
\end{subequations}
Using lemma \eqref{lem3} to estimate $\|s_a(t, \cdot)\|_{L^2}$ and defining
\begin{align*}
M(t, \cdot)= m \circ \sab- s_a \circ \overline m
\end{align*}
we have
\begin{subequations}
\begin{align}
&{\bf S}\left(s_e\right)= m \circ p_2 + \bar{p}_2 \circ m \label{newpair3}\\
&{\bf W}\left(\pb_2\right)= M +  m \circ \seb- s_e \circ \overline m \label{newpair4}\\
&\mbox{where}\notag\\
&\|M(t, \cdot)\|_{L^2(\mathbb R^6)} \le \frac{C \log(1+t)}{ 1+t^3}\notag
\end{align}
\end{subequations}
since
\begin{align*}
\|m \circ s\|_{L^2} \le C \|\phi\|^2_{L^{\infty}} \|s\|_{L^2} \le \frac{ C}{1+ t^3} \|s\|_{L^2}\ .
\end{align*}
Using the general formulas
\begin{align*}
&\W(s \circ \overline s) = \S(s) \circ \overline s - s \circ \overline{\S(s)}\\
&\W(p \circ p) = \W(p) \circ p + p \circ \W(p)
\end{align*}
and taking traces we get
\begin{align*}
&\frac{\partial}{\partial t}\left(\|s_e(t, \cdot)\|^2_{L^2(\mathbb R^6)} + \|p_2(t, \cdot)\|^2_{L^2(\mathbb R^6)}\right)\\
&\le \frac{C}{1+t^3} \left(\|p_2(t, \cdot)\|_{L^2(\mathbb R^6)}\|s_e(t, \cdot)\|_{L^2(\mathbb R^6)}\right) + C \|M(t, \cdot)\|_{L^2(\mathbb R^6)}\|p_2(t, \cdot)\|_{L^2(\mathbb R^6)}
\end{align*}
which leads to the desired estimate.
Explicitly,
define
\begin{align*}
E^2(t)=\|s_e(t, \cdot)\|^2_{L^2(\mathbb R^6)} + \|p_2(t, \cdot)\|^2_{L^2(\mathbb R^6)}
\end{align*}
then
\begin{align*}
\frac{\partial}{\partial t} E(t) \le C \left(\frac{1}{1+t^3} E(t) + \| M(t, \cdot)\|_{L^2}\right)
\end{align*}
thus $E(t)$ is uniformly bounded. Also, $\|p_2\|_{L^2}$ stays uniformly bounded and the logarithmic growth of $\|s_2\|_{L^2}$ can be traced back to
$\|s_a^{0}\|_{L^2}$ from
Lemma \eqref{s_a}.
\end{proof}

\section{List of all the error terms}

Our purpose in the present section is to compute explicitly all the error terms and show how they can be estimated.
Fortunately there are only a few terms for which one should be carefull, the rest being easier to handle.

Let us recall here briefly our basic strategy, which is to define
\begin{align*}
\big\vert\psi_{red}\big>:= e^{\B(k(t))}  e^{\sqrt{N}\A(\phi(t))}
e^{i t \H} e^{-\sqrt{N}\A(\phi(0))}\vac
\end{align*}
and compute $\H_{red}$ such that
$(1/i)\partial_{t}\big\vert\psi_{red}\big> =\H_{red} \big\vert\psi_{red}\big>$.
\begin{subequations}
\begin{align}
 &\H_{red}
=\frac{1}{i}\frac{\partial}{\partial t} \left(e^{\B(k(t))}
\right)e^{-\B(k(t))} \nonumber\\
&+e^{\B(k(t))}\left(\left(
\frac{1}{i}\frac{\partial}{\partial t}e^{\sqrt{N}\A(\phi(t))}\right)e^{-\sqrt{N}\A(\phi(t))}
  +e^{\sqrt{N}\A(\phi(t))}He^{-\sqrt{N}\A(\phi(t))} \right)e^{-\B(k(t))} \notag \\
&=N\mu_{0}(t) \label{zeroline}\\
&+N^{1/2}e^{\B(k(t))}
\int dx\left\{
h(\phi(t,x))a^{\ast}_{x}
+\bar{h}(\phi(t,x)) \right\}e^{-\B(k(t))} \label{firstline'}\\
&+
\frac{1}{i}\frac{\partial}{\partial t} \left(e^{\B(k(t))}
\right)e^{-\B(k(t)}
+  e^{\B(k(t))}
\bigg( \H_1 -\frac{1}{2}[[\A, [\A, \V]]]\bigg)e^{-\B(k(t))}\label{secondline}
 \\
 &-N^{-1/2}e^{\B(k(t))}\bigg( [\A, \V]+
 N^{-1/2} \V \bigg)e^{-\B(k(t))} \ .
\label{thirdline}
\end{align}
\end{subequations}
The function $\mu_{0}(t)$ and $h(\phi(t,x))$ appearing in \eqref{zeroline},\eqref{firstline'} are given below,
\begin{align*}
&\mu_{0} :=\int dx\left\{\frac{1}{2i}\big(\phi\bar{\phi}_{t}-\bar{\phi}\phi_{t}\big)-\big\vert\nabla\phi\big\vert^{2}
\right\}
-\frac{1}{2}\int dxdy\left\{v_{N}(x-y)\vert\phi(x)\vert^{2}\vert\phi(y)\vert^{2}\right\}
\\
&h(\phi(t,x)):=-\frac{1}{i}\partial_{t}\phi  +\Delta\phi -\big(v_{N}\ast\vert\phi\vert^{2}\big)\phi\ .
\end{align*}
As we know, \eqref{firstline'} is set to be zero due to the Hartree equation for $\phi$,
and \eqref{secondline} which contains terms of order $O(1)$  becomes block diagonal by an appropriate choice for the
evolution of $k$.
We can compute the $O(1)$ term in \eqref{secondline},
\begin{align*}
\eqref{secondline}=\H_{G}-
\I\left(\begin{matrix}
w^{T}&\overline{f}
\\
-f&-w
\end{matrix}
\right)
\end{align*}
where
\begin{align*}
&f:=\big(\S(\sh)-\chb\circ m\big)\circ\ch  -\big(\W(\chb)+\sh\circ \overline{m}\big)\circ \sh
\\
&w:=\big(\W(\chb)+\sh\circ\overline{m}\big)\circ \chb -\big(\S(\sh)-\chb\circ m\big)\circ\shb\ .
\end{align*}
The evolution implies that $f=0$ thus
\begin{align*}
\eqref{secondline}=\H_{G}-\I\left(\begin{matrix}
w^{T}&0\\ 0&-w
\end{matrix}
\right)\ .
\end{align*}
Multiplying $f$ on the right with $\shb$ and using the identities $\ch\circ\shb=\shb\circ\chb$ and $\sh\circ\shb =(\chb)^{2}-1$
we discover
\begin{align*}
f\circ\shb+w\circ\chb =\W(\chb)+\sh\circ\overline{m}\ .
\end{align*}
Using the formula (see \cite{GMM2}, lemma 3.1, stated in slightly different notation)
\begin{align*}
&\W(\chb) \circ \chb^{-1} +\chb^{-1} \circ \W(\chb)\\
& =\chb^{-1} \circ m  \circ \shb -\sh \circ \overline m \circ \chb^{-1}
\end{align*}
we get
\begin{align*}
w(y,x)&=\W\big(\chb\big)\circ \big(\chb\big)^{-1}+\sh\circ \overline{m}\circ\big(\chb\big)^{-1}=
\\
&\frac{1}{2}\Big(\big(\chb\big)^{-1}\circ m\circ \shb +\sh\circ\overline{m}\circ\big(\chb\big)^{-1}\Big)
+\frac{1}{2}\Big[\W(\chb),\big(\chb\big)^{-1}\Big]\ .
\end{align*}
If we write $w(y,x)$ for the above we have the quadratic term and a trace
\begin{align}
\eqref{secondline}=& \H_{G}+ \frac{1}{2} \Bigg(\int dxdy \left\{w(x,y)\D_{yx}+w(y,x)\D_{xy}\right\} \label{secondline-2}
\\
&+\frac{1}{2}{\rm Tr}\Big(\big(\chb\big)^{-1}\circ m\circ \shb +\sh\circ\overline{m}\circ\big(\chb\big)^{-1}\Big)\Bigg)\ .
\nonumber
\end{align}

The error terms, which are order $O(N^{-1/2})$ or higher, are all in the third line \eqref{thirdline} i.e. what we called $\E$,
\begin{align*}
\E=e^{\B}\left( [\A, \V]+N^{-1/2}\V\right) e^{-\B}\ .
\end{align*}
Recall that
\begin{subequations}
\begin{align}
&[\A, \V] =\int dx_{1}dx_{2}\left\{
 v_N(x_{1}-x_{2}) \left(\overline{\phi}(x_{2}) a^{\ast}_{x_{1}} \Q_{x_{1}x_{2}} + \phi(x_{2})\Q^{\ast}_{x_{1}x_{2}}a_{x_{1}}\right)
\right\} \label{Cubic-I}
\\
&\V =\frac{1}{2}\int dx_{1}dx_{2}\left\{v_{N}(x_{1}-x_{2})\Q^{\ast}_{x_{1}x_{2}}\Q_{x_{1}x_{2}}\right\}
\label{Quartic-I}
\end{align}
\end{subequations}
and the transformation of $(a,a^{\ast})$ is
\begin{subequations}
\begin{align}
&e^{\B} a_x e^{-\B} =
\int \left( \ch (y, x) a_y + \sh (y, x) a^*_y \right) dy :=b_{x}\label{b-I}\\
&e^{\B} a^*_x e^{-\B}=\int \left( \shb (y, x)a_y + \chb (y, x) a^*_y \right) dy:=b^{\ast}_{x}\ . \label{bstar-I}
\end{align}
\end{subequations}
One can see that $[b_{x},b^{\ast}_{y}]=\delta(x-y)$. For the rest of this section, the argument of the "trigonometric" functions is always $k$, thus $\shh$ abbreviates $\sh$, $p$ stands for $p_1$, etc.
Keeping in mind that $\chhb(y,x)=\chh(x,y)$ and $\shhb(y,x)=\shhb(x,y)$, $\shh(y,x)=\shh(x,y)$  we compute
\begin{subequations}
\begin{align}
&\ \ e^{\B}a^{\ast}_{x_{1}}a^{\ast}_{x_{2}}e^{-\B}=e^{\B}\Q^{\ast}_{x_{1}x_{2}}e^{-\B}= \nonumber \\
&\int dy_{1}dy_{2}\left\{
\shhb(y_{1},x_{1})\chh(x_{2},y_{2})\D_{y_{2}y_{1}}
+  \chhb(y_{1},x_{1})\shhb(x_{2},y_{2})\D_{y_{1}y_{2}}  \right\}+ \nonumber
\\
&\int dy_{1}dy_{2}
\left\{\chhb(y_{1},x_{1})\chh(x_{2},y_{2})\Q^{\ast}_{y_{1}y_{2}}
+\shhb(y_{1},x_{1})\shhb(x_{2},y_{2})\Q_{y_{1}y_{2}}
\right\} \nonumber
\\
&+\big(\shhb\circ\chhb\big)(x_{1},x_{2}) \label{star-left}
\\
&\ \ e^{\B}a_{x_{1}}a_{x_{2}}e^{-\B}=e^{\B}\Q_{x_{1}x_{2}}e^{-\B}=\nonumber \\
&\int dy_{3}dy_{4}
\left\{
\chh(y_{3},x_{1})\shh(x_{2},y_{4})\D_{y_{4}y_{3}}+\shh(y_{3},x_{1})\chhb(x_{2},y_{4})\D_{y_{3}y_{4}}
\right\}+ \nonumber
\\
&\int dy_{3}dy_{4}
\left\{
\chh(y_{3},x_{1})\chhb(x_{2},y_{4})\Q_{y_{3}y_{4}}+\shh(y_{3},x_{1})\shh(x_{2},y_{4})\Q^{\ast}_{y_{3}y_{4}}
\right\} \nonumber
\\
&+\big(\chhb\circ\shh\big)(x_{1},x_{2}) \ . \label{nostar-right}
\end{align}
\end{subequations}
Let us look first at the all
the terms that come from the quartic term $\V$. Since $e^{\B}\V e^{-\B}\vac$ we are interested only in terms that do not anihilate
 the vacuum.
This implies that we keep $\Q^{\ast}_{y_{3}y_{4}}$ from \eqref{nostar-right} with everything in \eqref{star-left} as well as $\Q^{\ast}_{y_{1}y_{2}}$
from \eqref{star-left} with the last term in \eqref{nostar-right}
and the product of the two last terms in \eqref{star-left}, \eqref{nostar-right}.

Below is the list of all the terms in   $\frac{1}{N}e^{\B}\V e^{-\B}$ (which do not annihilate $\vac$) in raw form,

\begin{subequations}
\begin{align}
&\frac{1}{2N}\quad \int dx_{1}dx_{2} dy_1 dy_2 dy_3 dy_4 \Big\{ \nonumber \\
&\shhb(y_{1},x_{1})\chh(x_{2},y_{2})
v_{N}(x_{1}-x_{2})
\shh(y_{3},x_{1})\shh(x_{2},y_{4})
\D_{y_{2}y_{1}}\Q^{\ast}_{y_{3}y_{4}}+ \label{quartic-a}
\\
&\chhb(y_{1},x_{1})\shhb(x_{2},y_{2})
v_{N}(x_{1}-x_{2})\shh(y_{3},x_{1})\shh(x_{2},y_{4})
\D_{y_{1}y_{2}}\Q^{\ast}_{y_{3}y_{4}} +\label{quartic-b}
\\
&\chhb(y_{1},x_{1})\chh(x_{2},y_{2})v_{N}(x_{1}-x_{2})\shh(y_{3},x_{1})\shh(x_{2},y_{4})
\Q^{\ast}_{y_{1}y_{2}}\Q^{\ast}_{y_{3}y_{4}} +\label{quartic-c}
\\
&\shhb(y_{1},x_{1})\shhb(x_{2},y_{2})v_{N}(x_{1}-x_{2})\shh(y_{3},x_{1})\shh(x_{2},y_{4})
\Q_{y_{1}y_{2}}\Q^{\ast}_{y_{3}y_{4}} \Big\} \label{quartic-d}
\\
&+\frac{1}{2N}\quad \int dx_{1}dx_{2}  dy_1 dy_2 \Big\{\notag \\
&\big(\shhb\circ\chhb\big)(x_{1},x_{2})v_{N}(x_{1}-x_{2})\shh(y_{1},x_{1})\shh(x_{2},y_{2})
\Q^{\ast}_{y_{1}y_{2}}+\label{quartic-e}
\\
&\chhb(y_{1},x_{1})\chh(x_{2},y_{2})v_{N}(x_{1}-x_{2})\big(\chhb\circ\shh\big)(x_{1},x_{2})
\Q^{\ast}_{y_{1}y_{2}}\Big\}\label{quartic-f}
\\
&+\frac{1}{2N}\quad \int dx_{1}dx_{2}\notag\\
&\big(\shhb\circ\chhb\big)(x_{1},x_{2})v_{N}(x_{1}-x_{2})\big(\chhb\circ\shh\big)(x_{1},x_{2}) .
\label{quartic-g}
\end{align}
\end{subequations}

Next let us look at the terms that come from the cubic term, namely the expression $e^{\B}[\A,\V ]e^{-\B}\vac$.
They come in two sets. First from the term $b^{\ast}_{x_{1}}(b_{x_{1}}b_{x_{2}})$ there are three terms, two coming from
the product of $\Q^{\ast}$ from \eqref{nostar-right} with \eqref{bstar-I} and one coming from $a^{\ast}$ in \eqref{bstar-I} with the
constant term in \eqref{nostar-right}. The are
listed below
\begin{subequations}
\begin{align}
&N^{-1/2}\int dx_{1}dx_{2} dy_1 dy_2 dy_3 \Big\{ \nonumber \\
&\shhb(y_{1},x_{1})v_{N}(x_{1}-x_{2})\bar{\phi}(x_{2})\shh(y_{2},x_{1})\shh(x_{2},y_{3})
\ a_{y_{1}}\Q^{\ast}_{y_{2}y_{3}}
+ \label{cubicI-a}
\\
&\chhb(y_{1},x_{1})v_{N}(x_{1}-x_{2})\bar{\phi}(x_{2})\shh(y_{2},x_{1})\shh(x_{2},y_{3})
\ a^{\ast}_{y_{1}}\Q^{\ast}_{y_{2}y_{3}} \Big\}
+ \label{cubicI-b}
\\
&N^{-1/2}\int dx_{1}dx_{2} dy_1\chhb(y_{1},x_{1})v_{N}(x_{1}-x_{2})\bar{\phi}(x_{2})\big(\chhb\circ\shh\big)(x_{1},x_{2})
\ a^{\ast}_{y_{1}}
  \label{cubicI-c}
\end{align}
\end{subequations}
The second set comes from $(b^{\ast}_{x_{1}}b^{\ast}_{x_{2}})b_{x_{1}}$
and gives five terms, namely all terms appearing in \eqref{star-left} multiplied with $a^{\ast}$ in \eqref{bstar-I}.
They are listed below,
\begin{subequations}
\begin{align}
&N^{-1/2}\int dx_{1}dx_{2} dy_1 dy_2 dy_3 \Big\{ \nonumber \\
&\shhb(y_{1},x_{1})\chh(x_{2},y_{2})
v_{N}(x_{1}-x_{2})\phi(x_{2})\shh(y_{3},x_{1})
\D_{y_{2}y_{1}}a^{\ast}_{y_{3}}
+ \label{cubicII-a}
\\
&\chhb(y_{1},x_{1})\shhb(x_{2},y_{2})
v_{N}(x_{1}-x_{2})\phi(x_{2})\shh(y_{3},x_{1})
\ \D_{y_{1}y_{2}}a^{\ast}_{y_{3}}
+ \label{cubicII-b}
\\
&\chhb(y_{1},x_{1})\chh(x_{2},y_{2})v_{N}(x_{1}-x_{2})\phi(x_{2})\shh(y_{3},x_{1})
\Q^{\ast}_{y_{1}y_{2}}a^{\ast}_{y_{3}}
+ \label{cubicII-c}
\\
&\shhb(y_{1},x_{1})\shhb(x_{2},y_{2})v_{N}(x_{1}-x_{2})\phi(x_{2})\shh(y_{3},x_{1})
\ \Q_{y_{1}y_{2}}a^{\ast}_{y_{3}} \Big\}
+ \label{cubicII-d}
\\
&N^{-1/2}\int dx_{1}dx_{2} dy_1\big(\shhb\circ\chhb\big)(x_{1},x_{2})v_{N}(x_{1}-x_{2})\phi(x_{2})\shh(y_{1},x_{1})
a^{\ast}_{y_{1}}
  \label{cubicII-e}
\end{align}
\end{subequations}
Some of the terms can be reduced to lower order by commuting $a$ with $a^{\ast}$ whenever they appear together in a product.
The term in \eqref{cubicII-d} applied to $\vac$ gives zero. Irreducible
terms are those appearing in \eqref{quartic-c} which is fourth order, as well as the terms appearing in \eqref{cubicI-b} and \eqref{cubicII-c}
which are cubic.
The quartic irreducible term in \eqref{quartic-c} can be writen, if we write $\chh(x,y)=\delta(x-y)+p(x,y)$ first, as follows,
\begin{subequations}
\begin{align}
&N^{-1}(1/2)\Big\{ \ \ v_{N}(y_{1}-y_{2})\shh(y_{3},y_{1})\shh(y_{2},y_{4})\quad + \label{main-quartic-irred}
\\
&\int dx
\left\{\bar{p}(y_{2},x)v_{N}(y_{1}-x)\shh(x,y_{4})\right\}\shh(y_{3},y_{1})+ \label{quartic-irred-1}
\\
&\int dx\left\{\bar{p}(y_{1},x)v_{N}(x-y_{2})\shh(y_{3},x)\right\}\shh(y_{2},y_{4})+ \label{quartic-irred-2}
\\
&\int dx_{1}dx_{2}
\big\{
\bar{p}(y_{1},x_{1})p(x_{2},y_{2})v_{N}(x_{1}-x_{2})\shh(y_{3},x_{1})\shh(x_{2},y_{4})\big\}\Big\}
\ . \label{quartic-irred-3}
\end{align}
\end{subequations}
The rest of the quartic terms can be reduced to either quadratic or zero order terms. For example if we look at
\eqref{quartic-b} we can write,
$\D_{y_{1}y_{2}}\Q^{\ast}_{y_{3}y_{4}}=\delta(y_{2}-y_{4})\Q^{\ast}_{y_{1}y_{4}}+
\delta(y_{2}-y_{4})\Q^{\ast}_{y_{1}y_{3}}$  (modulo terms which annihilate $\vac$)  and similarly for  \eqref{quartic-a}. Below is a list of all quadatic terms,
\begin{subequations}
\begin{align}
&(1/2N)\int dx_{1}dx_{2}\qquad\Big\{\nonumber \\
&\chhb(y_{1},x_{2})\shh(x_{2},y_{2})\big(\shhb\circ\shh\big)(x_{1},x_{1})v_{N}(x_{1}-x_{2}) + \label{quadratic-1}
\\
&\chhb(y_{1},x_{2})\shh(x_{1},y_{2})\big(\shhb\circ\shh\big)(x_{1},x_{2})v_{N}(x_{1}-x_{2}) +
\label{quadratic-2}
\\
&\chhb(y_{1},x_{1})\shh(x_{2},y_{2})\big(\shh\circ\shhb\big)(x_{1},x_{2})v_{N}(x_{1}-x_{2})
 + \label{quadratic-3}
\\
&\chhb(y_{1},x_{1})\shh(x_{1},y_{2})\big(\shh\circ\shhb\big)(x_{2},x_{2})v_{N}(x_{1}-x_{2}) +
\label{quadratic-4}
\\
&\shh(y_{1},x_{1})\shh(x_{2},y_{2})\big(\shhb\circ\chhb\big)(x_{1},x_{2})v_{N}(x_{1}-x_{2}) +
\label{quadratic-5}
\\
&\chhb(y_{1},x_{1})\chh(x_{2},y_{2})\big(\chhb\circ\shh\big)(x_{1},x_{2})v_{N}(x_{1}-x_{2})
\Big\}\ . \label{quadratic-6}
\end{align}
\end{subequations}
In addition \eqref{quartic-d} and \eqref{quartic-g} together with the trace in \eqref{secondline-2} will supply a zero order term,

\begin{align}
\mu_{1}=
&-\frac{1}{4}{\rm Tr}\Big(\big(\chb\big)^{-1}\circ m\circ \shb +\sh\circ\overline{m}\circ\big(\chb\big)^{-1}\Big)
\nonumber \\
&+\frac{1}{2N}\int dx_{1}dx_{2}\Big\{
(\shh\circ\shhb)(x_{1},x_{2})v_{N}(x_{1}-x_{2})(\shhb\circ\shh)(x_{1},x_{2})
+ \nonumber \\
&(\shh\circ\shhb)(x_{1},x_{1})v_{N}(x_{1}-x_{2})(\shhb\circ\shh)(x_{2},x_{2})+ \nonumber
\\
&\big(\shhb\circ\chhb\big)(x_{1},x_{2})v_{N}(x_{1}-x_{2})\big(\chhb\circ\shh\big)(x_{1},x_{2})
\Big\}\ . \label{mu-I}
\end{align}
This term can be absorbed as a phase factor in the evolution.
The cubic irreducible terms are those appearing in \eqref{cubicI-b},\eqref{cubicII-c}. Writing $\chhb(y, x)=\delta(y-x)+\bar{p}(y,x)$
we can express them in the manner below
\begin{subequations}
\begin{align}
&N^{-1/2}\Big\{ \quad\ v_{N}(y_{1}-y_{2})\phi(y_{2})\shh(y_{3},y_{1})\qquad\  \ + \label{main-cubic-irred}
\\
&\ \ \int dx\left\{v_{N}(y_{1}-x)\bar{\phi}(x)\shh(x,y_{3})\right\}\shh(y_{2},y_{1}) + \label{cubic-irred-1}
\\
&\ \ \int dx\left\{\bar{p}(y_{1},x)v_{N}(x-y_{2})\shh(y_{3},x)\right\}\phi(y_{2}) \ \ + \label{cubic-irred-2}
\\
&\ \ \int dx\left\{\bar{p}(y_{2},x)v_{N}(y_{1}-x)\phi(x)\right\}\shh(y_{3},y_{1}) \  \ + \label{cubic-irred-3}
\\
& \int dx_{1}dx_{2}\left\{\bar{p}(y_{1},x_{1})v_{N}(x_{1}-x_{2})\bar{\phi}(x_{2})\shh(y_{2},x_{1})\shh(x_{2},y_{3})\right\}+
\label{cubic-irred-4}\\
&\int dx_{1}dx_{2}\left\{
\bar{p}(y_{1},x_{1})p(x_{2},y_{2})v_{N}(x_{1}-x_{2})\phi(x_{2})\shh(y_{3},x_{1})\right\}\Big\}\ . \label{cubic-irred-5}
\end{align}
\end{subequations}
The rest can be reduced to linear or eliminated, for example \eqref{cubicII-d} can be eliminated.
For the cubic terms in the first set.
From the commutator relation $a_{y_{1}}\Q^{\ast}_{y_{2}y_{3}}=\delta(y_{1}-y_{2})a^{\ast}_{y_{3}}+\delta(y_{1}-y_{3})a^{\ast}_{y_{2}}$ (modulo $a$ terms which kill $\vac$)
we reduce them to linear and here is the list,
\begin{subequations}
\begin{align}
&N^{-1/2}\int dx_{1}dx_{2}\qquad\Big\{ \nonumber
\\
&\shh(y_{1},x_{2})\big(\shhb\circ\shh\big)(x_{1},x_{1})\bar{\phi}(x_{2})v_{N}(x_{1}-x_{2})\ +
\\
&\shh(y_{1},x_{1})\big(\shhb\circ\shh\big)(x_{1},x_{2})\bar{\phi}(x_{2})v_{N}(x_{1}-x_{2})\ +
\\
&\chhb(y_{1},x_{1})\big(\chhb\circ\shh\big)(x_{1},x_{2})\bar{\phi}(x_{2})v_{N}(x_{1}-x_{2})
\Big\}\ . \label{linear-from-cubicI}
\end{align}
\end{subequations}
From the cubic terms in the second list we have the commutators
$\D_{y_{1}y_{2}}a^{\ast}_{y_{3}}=\delta(y_{2}-y_{3})a^{\ast}_{y_{1}}$
and
$\D_{y_{2}y_{1}}a^{\ast}_{y_{3}}=\delta(y_{1}-y_{3})a^{\ast}_{y_{2}}$ (modulo $a$ terms)
hence from \eqref{cubicII-a},\eqref{cubicII-b} as well as from \eqref{cubicII-e} we obtain
\begin{subequations}
\begin{align}
&N^{-1/2}\int dx_{1}dx_{2}\qquad \Big\{ \nonumber
\\
&\chhb(y_{1},x_{1})(\shh\circ\shhb)(x_{1},x_{2})\phi(x_{2})v_{N}(x_{1}-x_{2})\ +
 \\
&\chhb(y_{2},x_{2})(\shh\circ\shhb)(x_{1},x_{1})\phi(x_{2}) v_{N}(x_{1}-x_{2})\ + \label{linear-from-cubic-extra}
 \\
&\shh(y_{1},x_{1})\big(\shhb\circ\chhb\big)(x_{1},x_{2})\phi(x_{2})v_{N}(x_{1}-x_{2})
\Big\}\ . \label{linear-from-cubicII}
\end{align}
\end{subequations}

\section{Estimates for the error terms}
In this section we prove
\begin{proposition} \label{errorest}
The following estimates for the error terms holds:
\begin{align*}
&N^{-1/2}\| e^{\B} [\A, \V]  e^{-\B} \vac\|_{\FF} \le  N^{\frac{3 \beta - 1}{2}} \frac{\log^3(1+t)}{t^{3/2}}\\
&N^{-1}\| e^{\B} \V  e^{-\B} \vac\|_{\FF} \le C N^{3 \beta - 1}\log^4(1+t)
\end{align*}
\end{proposition}
\begin{proof}
All estimates are straightforward, based on Corollary \eqref{s1est} and  Lemma \eqref{linf}, and we only include a few typical ones from each category.

Estimate for \eqref{main-quartic-irred}:
\begin{align*}
&N^{-1}\| v_{N}(y_{1}-y_{2})\sh(y_{3},y_{1})\sh(y_{2},y_{4})\|_{L^2(dy_1 dy_2 dy_3 dy_4)}  \\
&\le
 N^{-1} \|v_N\|_{L^{\infty}} \|\sh\|^2_{L^2} \le C N^{3 \beta - 1}\log^2(1+t)
 \end{align*}
 Estimate for \eqref{quartic-irred-3}:
 \begin{align*}
&N^{-1}\|\int dx_{1}dx_{2}
\bar{p_1}(y_{1},x_{1})p_1(x_{2},y_{2})v_{N}(x_{1}-x_{2})\shh(y_{3},x_{1})\shh(x_{2},y_{4})\|_{L^2(dy_1 dy_2 dy_3 dy_4)}\\
& \le N^{-1} \|v_N\|_{L^{\infty}} \|\sh\|^2_{L^2} \|p_1\|^2_{L^2}\le C N^{3 \beta - 1}\log^4(1+t)
\end{align*}
Estimate for \eqref{quadratic-6}, keeping only the $\delta$ contribution in $\ch$:
\begin{align*}
&N^{-1}\|\sh(y_{1},y_{2})v_{N}(y_{1}-y_{2})\|_{L^2(dy_1 dy_2)} \le
 \|v_N\|_{L^{\infty}} \|\sh\|_{L^2} \le C N^{3 \beta - 1}\log(1+t)
\end{align*}
Estimate for \eqref{main-cubic-irred}, :
\begin{align*}
&N^{-1/2}\| v_{N}(y_{1}-y_{2})\phi(y_{2})\sh(y_{3},y_{1})\|_{L^2(dy_1 dy_2 dy_3)}\\
&
\le \|\phi\|_{L^{\infty}}  \|v_N\|_{L^2}\|\sh\|_{L^2} \le C N^{\frac{3 \beta - 1}{2}} \frac{\log(1+t)}{t^{3/2}}
\end{align*}
Estimates for \eqref{linear-from-cubic-extra}:
\begin{align*}
&N^{-1/2}\|\int dx_{1}dx_{2}
p_2(y_{2},x_{2})(\sh\circ\shb)(x_{1},x_{1})\phi(x_{2}) v_{N}(x_{1}-x_{2})\|_{L^2(dy_2)}\\
&\le N^{-1/2}\|p_2\|_{L^2} \|\sh \circ \shb (x, x)\|_{L^1(dx)}\|\phi\|_{L^{\infty}}  \|v_N\|_{L^2}\\
&\le N^{-1/2}\|p_2\|_{L^2} \|\sh \|^2_{L^2}\|\phi\|_{L^{\infty}}  \|v_N\|_{L^2}\\
&\le C N^{\frac{3 \beta - 1}{2}} \frac{\log^3(1+t)}{t^{3/2}}
\end{align*}

\end{proof}


\section{Appendix} \label{app}
The purpose of this appendix is to make certain connections with our previous work \cite{GMM1}, \cite{GMM2} and to provide some
useful information.
We would like to explain first the Lie algebra isomorphism which is crucial in our work.
This was explained in \cite{GMM1} but we include it here for completeness and for the convenience of the reader.
Let us define first
\begin{align*}
\A_{x}:=\left(\begin{matrix}a_{x} \\a^{\ast}_{x}\end{matrix}\right)
\quad {\rm and}\quad {\bf f}(x):=\left(\begin{matrix}f_{1}(x)\\ f_{2}(x)\end{matrix}\right)
\quad {\rm and}\quad J=\left(\begin{matrix}0&-1\\ 1&0\end{matrix}\right)
\end{align*}
and using these form
\begin{align}
\A\big({\bf f}\big):=\int dx \left\{ \A^{T}_{x}{\bf f}(x)\right\}=
\int dx \left\{f_{1}(x)a_{x}+f_{2}(x)a^{\ast}_{x}\right\}\ . \label{A-form}
\end{align}
It is straightforward to check the commutation,
\begin{align}
\big[\A({\bf f} ),\A({\bf g})\big]&= \int dx\left\{ f_{1}(x)g_{2}(x)-f_{2}(x)g_{1}(x)\right\} \nonumber \\
&=-\int dxdy\left\{{\bf f}^{ T}(x)\delta(x-y)J{\bf g}(y)\right\}\ . \label{Symplectic-form}
\end{align}

For $L^2 (dx dy)$ kernel functions $d(t,x,y)$ and $k(t,x,y)$, $l(t,x,y)$ such that $k$ and $l$ are symmetric in the $(x,y)$ variables form the symplectic matrix
\begin{align}
S(d,k,l):=\left(\begin{matrix}d&k\\ l&-d^{T}\end{matrix}\right)\ . \label{symplectic-matrix}
\end{align}
We will write $S(x,y)$ when convenient, suppressing the $t$ dependence.
Next we define the map $\I :\rm {sp}(\mathbb C)\mapsto {\rm Quad}$ from the space of complex, $L^2$ symplectic matrices to quadratic expressions in $(a,a^{\ast})$
as follows :
\begin{align}
\I(S):=-\frac{1}{2}\int dxdy \left\{\A^{T}_{x}S(x,y)J\A_{y}\right\}\ . \label{Lie-isom-A}
\end{align}

\begin{theorem}
Let ${\bf f}(x)$ a vector function  and $\A({\bf f} )$ the corresponding expression \eqref{A-form}. Then the following commutations
relations hold
\begin{subequations}
\begin{align}
&\big[\I(S), \A({\bf f}\ )\big]=\A\big(S\circ{\bf f}\ \big) \label{commut-1}
\\
&e^{\I(S)}\A({\bf f} )e^{-\I(S)}=\A\big(e^{S}\circ{\bf f}\ \big) \label{commut-exp}
\end{align}
\end{subequations}
where (for example) $S\circ{\bf f}=\int dx\{S(x,y){\bf f}(y)\}$ etc. Formula \eqref{commut-exp} holds for
any complex symplectic  $S \in  \rm{sp}_c(\mathbb R) = (C^T)^{-1}\rm {sp}(\mathbb R) C^T$
where $C$ is the change-of-basis matrix
\begin{align*}
 C= \frac{1}{\sqrt 2}
 \left(\begin{matrix}I&-iI\\ I&iI\end{matrix}\right)
 \end{align*}
 and $\rm {sp}(\mathbb R)$ is the Lie algebra of real symplectic matrices. This condition insures that $e^S$ is unitary.
 \end{theorem}

\begin{proof}
The commutation relation in \eqref{commut-1} can be easily checked, but we point that it follows from \eqref{Symplectic-form}.
For any rank one quadratics we have using \eqref{Symplectic-form}
\begin{align*}
\big[\A({\bf f})\A({\bf g})\ ,\ \A({\bf h})\big]=
-\A\Big(\big({\bf g}{\bf f}^{ T}+{\bf f}{\bf g}^{\ T}\big)\circ J{\bf h}\Big)\ .
\end{align*}
Thus for any $R$ we have
\begin{align*}
\Big[ \int dxdy\big\{\A^{T}_{x}R(x,y)\A_{y}\big\}\ ,\ \A({\bf f} )\Big]=
-\A\Big(\big(R+R^{T}\big)\circ J{\bf f}\Big)\ .
\end{align*}
Now specialize to
$R=(1/2)SJ$ with $S\in {\rm sp}$ and use $S^{T}=JSJ$ to complete the proof.

For the second formula, introduce a complex parameter $t$ and take  $|\psi \big>$  a Fock space vector with finitely many non-zero components. It is trivial to check, using \eqref{commut-1}, that all $t$ derivatives of
$e^{t \I(S)}\A({\bf f} )e^{- t\I(S)} |\psi \big>$ and $\A\big(e^{t S}\circ{\bf f}\ \big) |\psi \big>$
agree when $t=0$, and both the left hand side and right hand side are analytic if $|t|$ is sufficiently small. Thus they agree
for all $t$ complex, sufficiently small. To take $t$ large, we have to restrict ourselves to real $t$ and use
the group properties of the unitary family $e^{tS}$.

A more formal (but convincing) "proof" follows from
\begin{align*}
e^{\I(S)}{\mathcal C} e^{-\I(S)}={\mathcal C} +[\I(S),{\mathcal C}]+\frac{1}{2!}\big[\I(S),[\I(S),{\mathcal C}]\big]+\ldots
\end{align*}

\end{proof}

\begin{theorem}
The map $\I :{\rm sp}(\mathbb C) \mapsto {\rm Quad}$ defined in \eqref{Lie-isom-A} is a Lie algebra isomorphism.
Moreover for $S, R\in {\rm sp_c(\mathbb R)}$ we have the formulas
\begin{subequations}
\begin{align}
&\I\Big(\big(\partial_{t}e^{S} \big)\circ e^{-S}\Big)=\big(\partial_{t}e^{\I(S)}\big)e^{-\I(S)} \label{commut-deriv-A}
\\
&\I\Big(e^{S}\circ R\circ e^{-S}\big)=e^{\I(S)}\I(R)e^{-\I(S)}\ . \label{commut-exp-matrix}
\end{align}
\end{subequations}

\end{theorem}

\begin{proof} We point out that \eqref{commut-exp} implies \eqref{commut-exp-matrix}, at least when $R$ is rank one,
$
R:={\bf f}(x){\bf g}^{T}(y)
$.
Notice that \eqref{commut-exp} can be written
\begin{align*}
e^{\I(S)}\int dx\left\{ {\bf f}^{T}(x)\A_{x}\right\}e^{-\I(S)}=
\int dxdy\left\{{\bf f}^{ T}(x)e^{S^{T}}\A_{y}\right\}
\end{align*}
from which we have
\begin{align*}
&e^{\I(S)}\int dxdy\left\{\A^{T}_{x}{\bf f}(x){\bf g}^{T}(y)J\A_{y}\right\}e^{-\I(S)} \\
&=
e^{\I(S)}\A({\bf f})e^{-\I(S)}\ e^{\I(S)}\int dy\left\{{\bf g}^{T}(y)J\A_{y}\right\}e^{-\I(S)} \\
&=\A\big(e^{S}{\bf f}\big)\int dydz\left\{{\bf g}^{T}(y)Je^{JSJ}\A_{z}\right\} \\
&=
\int dxdy \left\{\A^{T}_{x}e^{S}Re^{-S}J\A_{y}\right\}
\end{align*}
since $S^{T}=JSJ$ and $Je^{JSJ}=e^{-S}J$.

A direct proof of the Lie algebra  isomorphism follows from the following elementary computations. Define first
\begin{align*}
\Q_{xy}:=a_{x}a_{y}\quad {\rm and}\quad \Q^{\ast}_{xy}:=a^{\ast}_{x}a^{\ast}_{y}
\quad {\rm and}\quad \N_{xy}:=\frac{1}{2}\big(a_{x}a^{\ast}_{y}+a^{\ast}_{y}a_{x}\big)\ .
\end{align*}
It is straightforward to verify,
\begin{align*}
&\big[\Q_{x_{1}x_{2}},\Q^{\ast}_{y_{1}y_{2}}\big]=
\delta(x_{1}-y_{1})\N_{x_{2}y_{2}}+\delta(x_{1}-y_{2})\N_{x_{2}y_{1}} +
\delta(x_{2}-y_{1})\N_{x_{1}y_{2}}+\delta(x_{2}-y_{2})\N_{x_{1}y_{1}}
\\
&\big[\Q_{x_{1}x_{2}},\N_{y_{1}y_{2}}\big]
=\delta(x_{1}-y_{2})\Q_{x_{2}y_{1}}+\delta(x_{2}-y_{2})\Q_{x_{1}y_{1}}
\\
&\big[\N_{x_{1}x_{2}},\Q^{\ast}_{y_{1}y_{2}}\big]=
\delta(x_{1}-y_{1})\Q^{\ast}_{x_{2}y_{2}}+\delta(x_{1}-y_{2})\Q^{\ast}_{x_{2}y_{1}}
\\
&\big[\N_{x_{1}x_{2}},\N_{y_{1}y_{2}}\big]=
\delta(x_{1}-y_{2})\N_{y_{1}x_{2}}-\delta(x_{2}-y_{1})\N_{x_{1}y_{2}}
\end{align*}
and using these we can directly verify that
\begin{align*}
&\frac{1}{2}\left[\int dx_{1}dx_{2}\big\{k(x_{1},x_{2})\Q_{x_{1}x_{2}}\big\}
\ ,\ -\int dy_{1}dy_{2}\big\{l(y_{1},y_{2})\Q^{\ast}_{y_{1}y_{2}}\big\}\right]\\
&=-\int dxdy \big\{\big(k\circ l\big)(x,y)\N_{xy}\big\}
\end{align*}
which corresponds to the relation
\begin{align*}
\left[\left(\begin{matrix}0&k\\0&0\end{matrix}\right)\ ,\ \left(\begin{matrix}0&0\\ l&0\end{matrix}\right)\right]
=\left(\begin{matrix}k\circ l&0\\ 0&-l\circ k\end{matrix}\right)\ .
\end{align*}
The other cases are checked in a similar manner.

To prove \eqref{commut-deriv-A}, expand both the left and right-hand side as
\begin{align*}
&\I\left(\big(\partial_{t}e^{S}\big)e^{-S}\right) =
\I\left(\dot{S}+\frac{1}{2}[S,\dot{S}]+\ldots\right)\\
&=\dot{\I}(S)+\frac{1}{2}\big[\I(S),\dot{\I}(S)\big] +\ldots =
\big(\partial_{t}e^{\I(S)}\big)e^{-\I(S)}\ .
\end{align*}
The proof of \eqref{commut-exp-matrix} is along the same lines. The proofs can be made rigorous by an analyticity argument
on the dense subset of vectors with finitely many non-zero components.
\end{proof}

The second part of this appendix is devoted to the connection between the present equations for the pair excitations \eqref{pair-S-1}, \eqref{pair-W-1}
with the evolution equation derived in
\cite{GMM1} and \cite{GMM2}, the idea being that the nonlinear equation derived in these references turns out to be linear when
expressed in new coordinates. The connection however is not entirely obvious.
From the identities,
\begin{align*}
 e^{2 K}&=
 \left(
\begin{matrix}
\cht &\shbt\\
\sht & \chbt
\end{matrix}
\right)\\
&= \left(
\begin{matrix}
\ch \circ \ch + \shb \circ \sh &\ch \circ \shb+ \shb \circ \chb\\
\sh \circ \ch + \chb \circ \sh &\sh\circ \shb +  \chb \circ \chb
\end{matrix}
\right)\\
&=\left(
\begin{matrix}
2 \shb \circ \sh +1 & 2 \shb \circ \chb\\
 2 \chb \circ \sh &2 \sh\circ \shb +  1
\end{matrix}
\right)
\end{align*}
and
\begin{align*}
 I
= \left(
\begin{matrix}
\ch \circ \ch - \shb \circ \sh &-\ch \circ \shb+ \shb \circ \chb\\
\sh \circ \ch - \chb \circ \sh &-\sh\circ \shb +  \chb \circ \chb
\end{matrix}
\right)
\end{align*}
we have
$\chbt =2 \sh \circ \shb +1 = 2 \chb \circ \chb -1$ and $\sht= 2 \chb \circ \sh = 2 \sh \circ \ch$.
Using the above identities we can readily derive the identities,
\begin{align}
&{\bf W}\left(\sh \circ \shb\right)={\bf S}\left(\sh\right) \circ \overline{\sh} - \sh \circ  \overline{{\bf S}\left(\sh\right)}\label{W-id-1}   \\
&={\bf W}\left(\chb \circ \chb\right)={\bf W}\left(\chb\right) \circ \chb +\chb \circ {\bf W}\left(\chb\right) \nonumber\\
&{\bf S}\left(\chb \circ \sh\right)={\bf W}\left(\chb\right) \circ \sh + \chb \circ {\bf S} \left(\sh \right)\label{S-id-1}  \\
&= {\bf S} \left( \sh \circ \ch\right)={\bf S} \left(\sh \right) \circ \ch -\sh \circ\overline{ {\bf W}\left(\chb\right)}\ .\nonumber
\end{align}
The equations above have a generic character when we consider for example
$\S\big(\overline{f}\circ g\big)=\W\big(\overline{f}\big)\circ g+\overline{f}\circ \S(g)$ etc.

Our original evolution equation, see \cite{GMM2} was
\begin{align} \label{orignlsshort}
&\left(\frac{1}{i} \sh_t  + g^T\circ \sh +\sh \circ g  - \chb\circ m\right)\circ \ch  =\\
&\left(\frac{1}{i} \chb_t + [g^T, \chb] + \sh \circ\overline m\right)\circ\sh ~,\notag
\end{align}
or its equivalent form
\begin{align} \label{newnlsshort}
&\frac{1}{i} \sh_t  + g^T\circ \sh +\sh \circ g  - \chb\circ m  \\
&-\left(\frac{1}{i} \chb_t + [g^T, \chb] + \sh \circ\overline m\right)\circ(\chb)^{-1} \circ\sh=0~.\notag
\end{align}
The second equation can be written,
\begin{align*}
\S(\sh)-\W(\chb)\circ\big(\chb\big)^{-1}\circ\sh =\chb\circ m+\sh\circ\overline{m}\circ\big(\chb\big)^{-1}\sh \ .
\end{align*}
Multiplying on the left \eqref{newnlsshort} with $\big(\chb\big)^{-1}$ we obtain
\begin{align*}
&\big(\chb\big)^{-1}\circ \S(\sh)-\big(\chb\big)^{-1}\circ\W\big(\chb\big)\circ\big(\chb\big)^{-1}\circ\sh =
\\
&m+\big(\chb\big)^{-1}\sh\circ\overline{m}\circ\big(\chb\big)^{-1}\circ\sh
\end{align*}
and using the fact that
\begin{align*}
\big(\chb\big)^{-1}\W\Big(\chb\Big)\big(\chb\big)^{-1}=-\W\Big(\big(\chb\big)^{-1}\Big)
\end{align*}
we have in view of the identities \eqref{S-id-1} the simple equation,
\begin{align}
\S(\z)=m+\z\circ\overline{m}\circ\z \ , \label{z-equation}
\end{align}
where $\z :=\big(\chb\big)^{-1}\circ\sh=\sh\circ\big(\ch\big)^{-1}$. Now from \eqref{z-equation} we can readily derive the equation
below,
\begin{align}
\W\big(\sh\circ\shb\big)=m\circ\shb\circ\chb -\sh\circ\ch\circ\overline{m} \label{pair-W-3}
\end{align}
which is the same as \eqref{pair-W-2} and it already appeared in \cite{GMM2}. To derive an equation for $\sht$ notice that we can write,
\begin{equation*}
\sht =(\chb)^{2}\z +\z(\ch)^{2}
\end{equation*}
and subsequently compute using \eqref{S-id-1},\eqref{z-equation} and  \eqref{pair-W-3} in a tedious but straightforward manner,
\begin{align*}
\S\big(\sht\big) &=\W\big((\chb)^{2}\big)\z+(\chb)^{2}\S(\z)+\S(\z)(\ch)^{2}-\z\overline{\W\big((\chb)^{2}\big)}
\\
&=\big[m\circ\shb\circ\chb -\sh\circ\ch\circ\overline{m}\big]\circ(\chb)^{-1}\circ\sh
\\
&+(\chb)^{2}\circ m+(\chb)^{2}\circ(\chb)^{-1}\circ\sh\circ \overline{m}\circ(\chb)^{-1}\sh
\\
&+m\circ(\ch)^{2}+\sh\circ(\ch)^{-1}\circ\overline{m}\circ\sh\circ (\ch)^{-1}\circ (\ch)^{2}
\\
&-\sh\circ(\ch)^{-1}\circ\big[\overline{m}\circ\sh\circ\ch -\ch\circ\shb\circ m\big]
\\
&=m\circ\cht+\chbt\circ m\ .
\end{align*}
This equation is \eqref{pair-S-2}.

\end{document}